\documentclass[12pt]{article}
\usepackage{amsfonts}
\usepackage{amsmath,amssymb}
\usepackage{amsthm}
\usepackage{mathrsfs}
\usepackage{amsxtra}
\usepackage{amstext}
\usepackage{amscd}
\usepackage{latexsym}
\newtheorem{thm}{Theorem}[section]

\newtheorem{prop}[thm]{Proposition}

\newtheorem{rem}[thm]{Remark}
\def\Tr{{\rm T}}

\def\T{{\mathfrak C}_1}
\def\BL{{\mathfrak B}}
\def\D{{\mathcal D}}
\def\F{{\mathscr F}}
\def\G{{\mathscr G}}
\def\H{{\mathscr H}}
\def\K{{\mathcal K}}
\def\L{{\mathcal L}}

\def\R{\mathbb R}
\def\C{\mathbb C}
\def\Z{\mathbb Z}
\def\N{\mathbb N}

\begin{document}
\phantom{e}

\vspace{-10mm}

\noindent{\bf\Large On  the open  Dicke-type  model 
generated by an infinite-component vector spin} 

\setcounter{footnote}{0}
\renewcommand{\thefootnote}{\arabic{footnote}}

\vskip8mm

\noindent\textbf{\large Ryota Kyokawa${}^{\dagger}$, } 
\textbf{\large Hajime Moriya${}^{\dagger}$}   \ \ 
and \ \ 
\textbf{\large Hiroshi Tamura${}^{\dagger}$} 
\footnote{present address: Komatsu University, Shichoumachi nu 1-3, Komatsu 923-8511, Japan}
\medskip

\noindent{${}^{\dagger}$Graduate School of the Natural Science and  Technology, Kanazawa University, Kakuma-machi Kanazawa 920-1192, Japan

\vspace{1cm}

\vspace{1cm}

\noindent{\bf Abstract}\\
We consider an  open Dicke model comprising a single infinite-component 
vector spin and  a  single-mode  harmonic oscillator 
which are 
connected by  Jaynes--Cummings-type interaction between them.   
This open quantum model is referred to as the OISD 
 (Open Infinite-component Spin Dicke) model.  
The  algebraic structure of the OISD  Liouvillian is studied in terms of 
superoperators  acting on the space of  density matrices.
An  explicit invertible  superoperator 
(precisely, a completely positive trace-preserving map) is obtained  
that transforms the OISD  Liouvillian  into  a sum of two independent 
Liouvillians, one generated   by a  dressed  spin only,  
the other generated by  a  dressed harmonic oscillator only.  
The  time evolution 
 generated by the OISD  Liouvillian
 is shown to be asymptotically equivalent 
to that generated by 
  an adjusted   decoupled Liouvillian   with 
some  synchronized frequencies of  the spin and the harmonic oscillator. 
This asymptotic equivalence implies  that 
 the time evolution  of the OISD model  dissipates completely 
in the presence of any (tiny)  dissipation.

\bigskip

\noindent{\it Keywords\/}: open quantum dynamics, 
Liouvillian, decoupling, open Dicke model,  
infinite-component vector spin, synchronization
\section{Introduction} 
\label{sec:INT}

Investigating open quantum dynamical systems  described by 
the Gorini--Kossakowski--Sudarshan--Lindblad--Davies-type master equation 
\cite{K, GKS, LINDBLAD, D2} 
 is fundamentally important for researching quantum optics \cite{BP}. 
In general, it is difficult to understand the dynamical properties of
open quantum models, even qualitatively. However, some known  models have been studied to date. 
The simplest  but nevertheless important  
 examples are open harmonic oscillator models   
whose  dynamics can be analyzed   by various  methods; for example, see
\cite{AJP, CS, NVZ}.  
 Note that the (original) Dicke model \cite{DICKE54}  pertains to
non-trivial  interaction between light and matter 
 and has been subjected to intensive research;  
see  \cite{Garraway} and the references therein.
Needless to say,   open-type  Dicke models  as in 
 \cite{BSH}   give   another   prototype of  open quantum models, 
 thereby stimulating both theoretical and numerical investigations; 
for example, see \cite{KLINDER} and the  references therein.  
However, 
  open Dicke models are  difficult to analyze in general. 
The present paper gives  a concrete open Dicke model  to which
a rigorous approach  is possible.

In order  to explain  our motivation 
 in  a broad context, we pose the following question. 
 To understand  a complicated quantum model, 
 what should we do first? 
Not in every case, but sometimes it is useful to  transform   it 
 into a simpler model. 
 As a special  realization of this simple  idea, 
 we recall Bogoljubov's  method, which 
 has been used widely in condensed matter physics and quantum field theory 
\cite{BOG}.
Bogoljubov's   method  is  tightly  
  related to the   concept of quasiparticles,  
and it is   essential  for  many important physics models such as 
Bardeen--Cooper--Schrieffer models and polaron models  \cite{FW}. 
However, although Bogoljubov's method 
has  many  applications as mentioned above,  it seems  that its 
roles   in  open quantum systems are yet to be explored fully: 
Let us mention a recent work \cite{PROSEN} that treats 
  the dynamics of open  fermion models.   
On that basis, we seek concrete   open quantum models 
to which  a Bogoljubov-like  method   can apply.   
In fact,  our  open  Dicke  model   
(mentioned above and specified  later) 
is    such  an  example,  
 although it is  somewhat artificial  as a physics model.     
   
 In our  Bogoljubov-like  method  
applied to open dynamics, we make use of 
 a completely positive and trace preserving (CPTP) map 
for its similarity transformation.  
Although  our argument is limited to this  particular model only, 
we  advocate a
 general  scheme for dealing with   open quantum dynamics as follows: 
 To  simplify (e.g., decouple,  diagonalize) 
Lindblad operators of open quantum systems, 
we  make use of  CPTP maps
for  similarity transformations rather than  
 unitary maps  used  in closed quantum systems.

We specify  our open-type  Dicke model.
First, its radiation field is given by   a  single-mode  
harmonic oscillator as in  the original Dicke model \cite{DICKE54}.  
Second,   its  matter is given by 
 a single infinite-component vector spin. 
Third, we  add  dissipation  in  the following way: 
the harmonic oscillator is implicitly connected to a thermal bath, 
whereas the matter  may or may not be connected to a thermal bath.    
Finally,  assuming   a Jaynes--Cummings-type interaction between   
the  harmonic oscillator and the spin,  we 
obtain our   open infinite-spin Dicke model.  
Hereinafter for abbreviation, we refer to this model as the OISD model.     

We provide  some   background to the OISD model.   
First, we  note   that   infinite-component vector spins 
can be extracted from $n$-compound  quantum  spin 
systems by a certain limiting procedure $n \to \infty$.   
The infinite-component spin Dicke  model, ISD for short, is   
a  {\it{conserved (non-dissipative)}}  quantum model 
that is  generated by  an infinite-component 
vector spin and a single-mode harmonic oscillator as described above.
It has been  investigated in  \cite{HL, D1, BZT}.   

Notably, the OISD model has   both   
dissipation (induced by a hidden thermal bath) and  
non-trivial interaction (between  the matter  and the radiation).       
We investigate their interplay and effects upon  the open  
quantum dynamics.
Through the  decoupling transformation,  we  delete  the interaction  
  producing   new dissipation terms in  both the spin and the harmonic 
oscillator.  Under this transformation, 
 no dissipation term  in the form of a mixture of the spin and 
the harmonic oscillator appears.

This paper is divided into several  parts as follows. 
In \S\ref{sec:HO}, we  discuss the algebraic  properties of 
the  open harmonic oscillator model in  a  self-contained manner.  
Although  some of our results 
in that section are either obvious or  
easily derived from known results  \cite{CS, BE, EFS, HNY, TP},  
our  reformulation  based on the algebra of superoperators  
is  useful for gaining information about the OISD model as well.   
In  \S\ref{sec:SPIN}, we introduce the infinite-component spin system. 
We show that dissipation (even a tiny amount) added to the infinite-component spin  makes the system thoroughly  unstable.      
The  special  algebraic structure   of the 
infinite-component  spin renders the OISD model tractable by purely 
mathematical methods (without resorting to numerical computations).    
In \S\ref{sec:OISD}, we define the OISD model  
and investigate its Liouvillian based on  \S\ref{sec:HO} and \S\ref{sec:SPIN}.  
In Proposition~\ref{prop:decom}, we show that the OISD  model is decomposed into a sum of  
two  {\it{independent}}   Liouvillians by a similarity transformation 
generated  by  a CPTP  superoperator (and its inverse).      
Here, independence means that  one is   written  by a  dressed  spin only,  
while   the other is written   by  a dressed harmonic oscillator only.     
As a result of the decomposition, the dynamics of 
the OISD model are examined. Proposition~\ref{prop:AYSMP}
 shows  that 
 the time evolution of the OISD model is 
asymptotically equivalent to that  of 
an   decoupled Liouvillian.
This   decoupled Liouvillian is adjusted so that it has   
 synchronized frequencies of  the spin and the harmonic oscillator.  
See \cite{GCZ}  for 
general information about synchronization for open quantum dynamics.         
Note that the decoupled Liouvillian in 
Proposition~\ref{prop:decom} and the adjusted decoupled Liouvillian 
 in Proposition~\ref{prop:AYSMP} differ.
Furthermore, in Proposition~\ref{prop:AYSMP} it is shown that in the presence of any (tiny) dissipation, the time evolution  of the OISD model dissipates completely as time tends  to infinity.

We now   introduce the mathematical notation 
that is used throughout this paper.    
Let $\H$ be an infinite-dimensional  Hilbert space.   
Let $\BL(\H)$ be the set of all bounded linear operators on $\H$   
and $\T(\H)$ be the Banach algebra of the trace class operators on $\H$ 
with trace norm $\| \, \cdot \,\|_1$.   
For  linear operators $A$, $B$, and $C$ acting in  
$\H$, we define superoperators $\K_A$ and 
$\D_{B\circ C}$ acting in $\T(\H)$ by   
\begin{equation}
\K_A(\rho): = [A, \rho], \quad
   \D_{B\circ C}(\rho): = 2B\rho C - \{CB, \rho \} 
   \quad \mbox{ for } \quad \rho \in \T(\H).  
\label{KDdef}
\end{equation}
We can verify   the following  relations by 
  straightforward calculations:   
\begin{align}
  [\K_A, \K_B] & =  \K_{[A, B]},  \quad 
   [\K_A, \D_{B \circ C}] = \D_{[A, B] \circ C}
   + \D_{B \circ [A, C]} ,
\label{KK}\\
  [\D_{A \circ B}, \D_{ C\circ D}] & =  
 \D_{[A,C] \circ\{B,D\}} - \D_{\{A, C\} \circ [B, D]} + 
 \D_{[DC,A] \circ B}  
{\notag}\\
 & + \D_{A \circ [B,DC]}   
  - \D_{[BA,C] \circ D} - \D_{ C\circ[D,BA] }
  + \K_{[BA, DC]} ,
\label{DD}\\
   \D_{ A\circ \bold{1}} &= \K_A , \quad
  \D_{ \bold{1}\circ B} = - \K_B .
\label{DK}
\end{align}

We investigate the   open quantum dynamics on $\H$ 
determined by the following type of master equation  
for $\rho(t)\in \T(\H)$:
\begin{equation}
     \frac{d}{dt}\rho(t) = -i\K_{H}(\rho(t)) 
     + \sum_j \D_{A_j\circ A_j^{\dagger}}(\rho(t)) \, , 
\label{MASTERdef}
\end{equation}
 where  $\K_H$ represents the infinitesimal change induced by 
a Hamiltonian $H$ of the corresponding closed quantum system, 
 and  the term $\sum_j  \D_{A_j\circ A_j^{\dagger}}$ 
with a set of    operators  $\{A_j$, $A_j^{\dagger}\}$ 
represents the effect of dissipation.

We make some remarks about the  paper.
We assume (implicitly) that superoperators act on  the  
 space of trace class operators rather than    
 on the space  of  Hilbert--Schmidt operators; 
    in some   literature  of mathematical  physics,
  the Hilbert--Schmidt class  seems to be  more common. 
 Throughout this paper, linear operators  in  a Hilbert space and 
 superoperators  are generically  unbounded.   
   From a mathematics standpoint,  our  result is
 a realization of a particular  open quantum dynamics with 
unbounded generators. 
We refer to \cite{SHW}, which suggests general  investigation of 
{\it{unbounded}}  open quantum dynamical systems.
Nevertheless, 
 we put more emphasis  on explicit calculations 
 of   this particular  model rather 
than on general arguments based on  functional analysis. 

\section{Algebraic structure of open harmonic oscillator} 
\label{sec:HO}
Because of the fundamental  importance of the open harmonic oscillator in 
open quantum systems, 
it is worth considering yet another  formalism that is easier to handle.    
We give such  in terms of the algebra of 
superoperators on the space of trace class operators 
on the one-mode Fock space.

Let $a$ and  $a^{\dagger}$ denote  
the annihilation and creation operators, respectively, acting  in  
 the one-mode Fock space  $\F$.   
The number operator is denoted as $N=a^{\dagger} a$.
It is known that the Fock space  $\F$ has the following  complete 
orthonormal {{system}}
$\{ \, |n \, \rangle\, \}_{n\in \N\cup\{0\}}$,  
such that
\begin{equation}
    N | n \rangle = n  | n \rangle \qquad (n \in \N\cup\{0\}), 
\label{Number}
\end{equation}
and 
\begin{equation}
    a |0\rangle =0, \quad  a^{\dagger} |0\rangle =|1 \rangle, 
\quad   
    a\, |  n \rangle = \sqrt{n} \, | n-1 \rangle, 
    \quad a^{\dagger}\, |  n \rangle = \sqrt{n+1} 
 \, | n+1 \rangle \quad (n \in \N). 
\label{hobase}
\end{equation}
Let us consider the master equation
\begin{equation}
     \frac{d}{dt}\rho(t) = \L_{\text{ph}}(\rho(t)) \qquad \mbox{ for }
     \quad \rho(t) \in \T(\F) 
\label{master0}
\end{equation}
with its Liouvillian 
\begin{equation}
   \L_{\text{ph}} :=  -i\omega\K_{N} 
  + \gamma\Big((J+1)\D_{a\circ a^{\dagger}} + J\D_{a^{\dagger}\circ a} \Big), 
\label{LIOU-ho}
\end{equation}
where $\omega>0$ denotes the angular frequency of the 
 oscillator, $\gamma > 0$ denotes the strength of dissipation,
and $J \geqslant 0$ is a parameter related to the temperature (of the hidden thermal bath) \cite{BP}.  
The subscript of $\L_{\text{ph}}$ signifies ``photon.''

We consider the algebraic structure of the Liouvillian $\L_{\text{ph}}$ and its constituents.   
From (\ref{KK}-\ref{DK}) and $[a, a^{\dagger}]  = \bold{1}$, we have
\begin{align}
    [\K_a, & \,\K_{a^{\dagger}}] = 0 \,, \quad
    [\K_a, \K_{N}] = \K_a \,, \quad
    [\K_{a^{\dagger}}, \K_{N}] = 
    - \K_{a^{\dagger}} \,,
\label{comKK} \\
    [\K_{a}, &\D_{a\circ a^{\dagger}}] = \K_a  \,, \quad
    [\K_{a^{\dagger}}, \D_{a\circ a^{\dagger}}] =
     \K_{a^{\dagger}} \,,  
\label{comKD1}
\\
    [\K_{a}, &\D_{a^{\dagger}\circ a}] = -\K_{a} 
\,, \quad 
    [\K_{a^{\dagger}}, \D_{a^{\dagger}\circ a}] =
    -\K_{a^{\dagger}} 
\label{comKD2}
\end{align}
and
\begin{equation}
    [\K_{N}, \D_{a\circ a^{\dagger}}] =
    [\K_{N}, \D_{a^{\dagger}\circ a}] = 0 \,,
\label{comKD3}
\end{equation}
\begin{equation}
    [\D_{a\circ a^{\dagger}}, \D_{a^{\dagger}\circ a}]
  = -2 \big( \D_{a\circ a^{\dagger}} + 
   \D_{a^{\dagger}\circ a}\big) \,.
\label{comDD}
\end{equation}

We then  investigate  the eigenvalue problem of the  Liouvillian (\ref{LIOU-ho}) using the algebraic relations listed  above.  
For convenience, we  define $ \Phi_{n,m} = \K_{a^{\dagger}}^n 
\K_a^m(|0\rangle \langle 0|)$ for $n, m = 0, 1, 2, \cdots$. 
We see that  the set $\{ \Phi_{n, m} \}_{n,m=0}^{\infty}$ 
 is total  in  $\T(\F)$, namely  its linear span is a dense subspace of 
$\T(\F)$.    
In fact, for each non-negative integer $k$, the identity  
Lin.Span$\{ \Phi_{n,m}  \, |  \, n \geqslant 0, \, k\geqslant m \geqslant 0 \} = $Lin.Span$\{ \, |n\rangle \langle m| \, | \, n \geqslant 0, k\geqslant m \geqslant 0 \}$ holds. 
We can show  this set of  identities   by induction on $k$ noting 
 $\Phi_{n,0} = \sqrt{n!}\,|n\rangle \langle 0|$ and $ \sqrt{m+1} \, |n\rangle\langle m+1| = \sqrt{n} \, |n-1\rangle\langle m| - \K_a(|n\rangle\langle m|) $. 
Using  (\ref{comKK}) and (\ref{comKD1}) inductively, we have 
\begin{equation}
    \K_{N} (\Phi_{n, m})
    = (n-m) \Phi_{n, m} \,, \qquad 
    \D_{a\circ a^{\dagger}} (\Phi_{n, m})
    = -(n+m) \Phi_{n, m} \qquad 
    ( n,m \in \N \cup \{ 0 \} \, ) \, .
\label{KDeigen}
\end{equation}
Thus  all  $\Phi_{n,m}$ are common eigenvectors for both  
$\K_{N}$  and $\D_{a\circ a^{\dagger}}$,  
 but  not  for $\D_{ a^{\dagger}\circ a}$.   
In the next paragraph, 
 we   focus on   $\D_{a^{\dagger}\circ a}$.  

We investigate the semigroup  generated by $\D_{a^{\dagger}\circ a}$. Set 
\begin{equation}
   S_t(\rho)  =
   \sum_{n=0}^{\infty}\frac{(1-e^{-2t})^n}{n!}
   a^{\dagger n}e^{-taa^{\dagger}} \rho \, 
   e^{-taa^{\dagger}} a^n  
\label{SEMI-aadag}
\end{equation}
for $ t \geqslant 0$ and $ \rho \in \T(\F)$.   
It is straightforward to derive the following properties: 
\begin{equation}
     \frac{d}{dt}S_t(\rho) = 
    \D_{a^{\dagger}\circ a} (S_t(\rho)) \,, \qquad
     S_{t_1+t_2}(\rho) = S_{t_1}\big(S_{t_2}(\rho)\big) \,,  \qquad   S_0(\rho) = \rho \,. 
\label{3eq}
\end{equation}
The semigroup $S_t$ defined above 
is CPTP 
because by  definition (\ref{SEMI-aadag}) 
 it   has the following  
form of  CPTP maps (e.g., see \cite{NC} \S\S8.2.4, \cite{AL}):
\begin{equation}
S_t (\rho) = \sum_{n=0}^{\infty} E_n(t)\rho E_n^{\dagger}(t) 
\label{CP}
\end{equation}
with  bounded operators $E_n(t)$ satisfying the normalization condition
\begin{equation}
    \sum_{n=0}^{\infty} E_n^{\dagger}(t) E_n(t) =  \bold{1}    
\label{trace=1}
\end{equation}
for every  $ t \geqslant 0$.  
Precisely, conditions (\ref{CP}) and (\ref{trace=1}) are satisfied by 
\begin{equation}
 E_n(t) =  \frac{(1-e^{-2t})^{n/2}}{\sqrt{n!}}   a^{\dagger n}e^{-taa^{\dagger}}\in\BL(\F)  \,.
\label{CPgutai}
\end{equation}
Similarly
$\D_{a \circ a^{\dagger}}$ generates a one-parameter semigroup of 
CPTP maps  on $\T(\F)$, as  the following formula holds:   
\begin{equation}
   e^{t\D_{a \circ a^{\dagger}}} (\rho) =
   \sum_{n=0}^{\infty}\frac{(e^{2t} - 1)^n}{n!}
   a^ne^{-ta^{\dagger}a} \rho \, 
   e^{-ta^{\dagger}a} a^{\dagger n}  \,.  
\label{semigroup2}
\end{equation}

We show the following identity of  superoperators:  
\begin{equation}
   e^{\tau\D_{a^{\dagger}\circ a}}  
  \D_{a\circ a^{\dagger}} = 
  \big(e^{2\tau} \D_{a\circ a^{\dagger}} 
   +(e^{2\tau}-1)   \D_{a^{\dagger}\circ a} \big)
   e^{\tau\D_{a^{\dagger}\circ a}} \qquad ( t\geqslant 0)\,. 
\label{transform0}
\end{equation}
Denote  the right-hand side of (\ref{transform0}) by  $X_{\tau}$.   
Then  we   immediately see that 
$ X_0 = \D_{a\circ a^{\dagger}}$ and  $
     \frac{d}{d\tau}X_{\tau} = 
       \D_{a^{\dagger}\circ a}X_{\tau}$  
because of  (\ref{comDD}).

Setting $\tau \geqslant 0$ by  
\begin{equation}
  e^{2\tau} = J + 1
\label{tau-KIGOU}
\end{equation}
together with (\ref{comKD3}), we have
\begin{equation}
   e^{\tau\D_{a^{\dagger}\circ a}} 
   \big(-i\omega\K_{N} +   
   \gamma \D_{a\circ a^{\dagger}} \big) =      
    \L_{\text{ph}} \, e^{\tau\D_{a^{\dagger}\circ a}}  \, . 
\label{transform1}
\end{equation}
Because of (\ref{KDeigen}), (\ref{transform0}), and (\ref{transform1}),
 we obtain  
\begin{equation}
\Big((J+1)\D_{a\circ a^{\dagger}} + J\D_{a^{\dagger}\circ a} \Big)  
 e^{\tau\D_{a^{\dagger}\circ a}}(\Phi_{n,m})
  = - (n+m)
  e^{\tau\D_{a^{\dagger}\circ a}} (\Phi_{n,m})
  \qquad (n,m \in \N \cup \{ 0 \}) 
\label{JD0eigen}
\end{equation}
and  
\begin{equation}
  \L_{\text{ph}} e^{\tau\D_{a^{\dagger}\circ a}}(\Phi_{n,m})
  =\big( -i\omega(n-m) - \gamma(n+m)\big)
  e^{\tau\D_{a^{\dagger}\circ a}} (\Phi_{n,m})
  \qquad (n,m \in \N \cup \{ 0 \}) \,.
\label{L0eigen}
\end{equation}

To see that this gives 
 the complete solution of the eigenvalue problem 
for  $\L_{\text{ph}}$, it is enough to show that the range of 
$ \, e^{\tau\D_{a^{\dagger}\circ a}} \, $  is dense in 
$\T(\F)$, because 
$\{  \Phi_{n,m} \}_{n,m=0}^{\infty}$ 
is total in $\T(\F)$. 
We note that the right-hand side of (\ref{SEMI-aadag}) 
is still  well defined for $ t < 0$ and gives  
a densely defined operator as the (left) inverse of 
$ \, e^{|t|\D_{a^{\dagger}\circ a}} \, $
when $|t|$ is small (i.e., $|1-e^{-2t}| < 1$).  
Hence the range of $ \, e^{\tau\D_{a^{\dagger}\circ a}} \, $  
is dense in $\T(\F)$ for small $\tau >0\,$. 
The bounded semigroup property ensures the same for 
arbitrary $\tau >0$. 
We refer to 
previous work  \cite{CS, BE, EFS,HNY}
for  the information regarding the eigenvalue problem of   
$\L_{\text{ph}}$.

Let us consider the semigroup $e^{t\L_{\text{ph}}}$.  
Because of the commutativity (\ref{comKD3}),  we have 
\begin{equation}
      e^{t\L_{\text{ph}}}  =  e^{-it\omega\K_N}e^{t\gamma\big((J+1)\D_{a\circ a^{\dagger}}
      + J \D_{a^{\dagger}\circ a}\big)}.
\label{evolutionO} 
\end{equation}
We have also  
\begin{align}
 e^{t\gamma\big((J+1)\D_{a\circ a^{\dagger}}
      +J \D_{a^{\dagger}\circ a}\big)}
     =   e^{\tau_1(t\gamma)\D_{a^{\dagger}\circ a}}
      e^{(t\gamma+\tau_1(t\gamma))\D_{a\circ a^{\dagger}}} \,,
\label{evolutionD}
\end{align}
where
 \begin{equation}
          \tau_1(s) := \frac{1}{2}\log(J+1-Je^{-2s})       
           \qquad \mbox{for} \quad s\geqslant 0 \,. 
\label{tauone}
\end{equation}
To see (\ref{evolutionD}), we set its right-hand side as $Y(t)$.   
Then $Y(0) = \bf{1}$ holds  and its derivative satisfies  
\begin{align*}
       \frac{dY(t)}{dt} =& \gamma\tau_1'(t\gamma) \D_{a^{\dagger}\circ a}
       e^{\tau_1(t\gamma)\D_{a^{\dagger}\circ a}}
      e^{(t\gamma+\tau_1(t\gamma))\D_{a\circ a^{\dagger}}} 
\\
   & + e^{\tau_1(t\gamma)\D_{a^{\dagger}\circ a}}  \gamma\big( 1+\tau_1'(t\gamma) \big)
     \D_{a\circ a^{\dagger}}  
      e^{(t\gamma+\tau_1(t\gamma))\D_{a\circ a^{\dagger}}}       
\\
    =& \gamma\big(\tau_1'(t\gamma) \D_{a^{\dagger}\circ a} + 
      \big( 1 + \tau_1'(t\gamma)\big)(e^{2\tau_1(t\gamma)} \D_{a\circ a^{\dagger}} 
      + (e^{2\tau_1(t\gamma)} - 1)\D_{a^{\dagger}\circ a})\big) Y(t)
\\
     =& \gamma\big((J+1)\D_{a\circ a^{\dagger}} + J\D_{a^{\dagger}\circ a}\big) Y(t)  \,, 
\end{align*}
where  we have used (\ref{transform0}).   
With (\ref{evolutionO}),
 (\ref{evolutionD}), and (\ref{tauone}),
we  obtain  the following decomposition formula of $e^{t\L_{\text{ph}}}$: 
\begin{equation}
    e^{t\L_{\text{ph}}}         
 =  e^{-it\omega\K_N} e^{\tau_1(t\gamma)\D_{a^{\dagger}\circ a}}
      e^{(t\gamma+\tau_1(t\gamma))\D_{a\circ a^{\dagger}}} \,.
\label{decomposition}
\end{equation}
Because  $ e^{-it\omega K_{N} }$  (which generates   a unitary evolution) is 
 obviously a CPTP semigroup on $\T(\F)$, 
and both  $e^{t\D_{a^{\dagger}\circ a}}$
and $e^{t\D_{a\circ a^{\dagger}}}$ 
are CPTP maps as we have seen,  
$e^{t\L_{\text{ph}}}$  given as the composition of 
these CPTP maps in (\ref{decomposition}) is also  a 
 CPTP map.

In the following, we discuss  the asymptotic behavior of 
 the dynamical semigroup  $\{ e^{t \L_{\text{ph}}} \}_{t\geqslant 0}$.  
From (\ref{SEMI-aadag}) and (\ref{tau-KIGOU}), we have  
\begin{equation}
    e^{\tau\D_{a^{\dagger}\circ a}} \,(\Phi_{0,0}) 
    = S_{\tau}(\Phi_{00}) 
  = \sum_{n=0}^{\infty}\frac{e^{-\beta\omega n}}{Z}
    |\, n\rangle\langle n\,|  \equiv \rho_{\text{ph}, \text{G}}, \qquad \beta = \omega^{-1}\log(1+ J^{-1}). 
\label{Gibbs0}
\end{equation}
The   right-hand  side of the above equality is 
the Gibbs state with respect to 
the Hamiltonian $\omega N$ at  the inverse temperature $\beta$.
It follows from (\ref{L0eigen})  that $\rho(t)$ 
approaches the Gibbs state (\ref{Gibbs0}) in 
trace norm as $ t \to \infty$ from arbitrary initial state $\rho(0)$.   
To see this asymptotic  property, we  note that the trace-preserving 
property of $ e^{\tau\D_{a^{\dagger}\circ a}} $ and the trace property 
yield
\begin{equation}
    \Tr[e^{\tau\D_{a^{\dagger}\circ a}}(\Phi_{n, m})] = 
    \Tr[\Phi_{n, m}]  = \delta_{n, 0}\delta_{m, 0} \qquad 
    ( n,m \in \N \cup \{ 0 \} ) \, .
\label{nm}
\end{equation}
Because  the trace-preserving 
map $e^{\tau\D_{a^{\dagger}\circ a}}:\T(\F)\rightarrow \T(\F)$ 
has  dense range with respect to the trace norm,  
any density matrix $\rho_0 $ is 
approximated by the finite linear combination      
\begin{equation*}
  \sum_{n, m \geqslant 0} c_{n, m} 
     e^{\tau\D_{a^{\dagger}\circ a}}(\Phi_{n, m}), \quad c_{n,m}\in \C      
\end{equation*}
in $\T(\F)$.    
Note that  $ c_{0, 0} =1$ because of 
$\Tr(\rho_0)=1$ and  (\ref{nm}).  
From (\ref{L0eigen}) and (\ref{Gibbs0}), we obtain 
\begin{equation}
      \lim_{t \to \infty} e^{t\L_{\text{ph}}}\big(\rho_0\big)  =   
       \sum_{n, m \geqslant 0}c_{n, m}\lim_{t \to \infty} 
   e^{t\L_{\text{ph}}}\big(e^{\tau\D_{a^{\dagger}\circ a}}(\Phi_{n, m})\big)
 = \ e^{\tau\D_{a^{\dagger}\circ a}} (\Phi_{0, 0})
 = \ \rho_{\text{ph}, \text{G}}  \,.  
\label{toGibbs}
\end{equation}
The above heuristic  limiting procedure can be made rigorous 
by using the uniform boundedness of CPTP semigroup 
$ \{ \,e^{t\L_{\text{ph}}}\,\}_{t \geqslant 0}$.     
From  (\ref{toGibbs}) it follows that 
 for any $ t\geqslant 0$,
\begin{equation}
       e^{t\L_{\text{ph}}}\big(\rho_{\text{ph}, \text{G}}\big) 
 = \rho_{\text{ph}, \text{G}}  \,.  
\label{INVGibbs}
\end{equation}
Recall that  the Gibbs state $\rho_{\text{ph}, \text{G}}$
is with respect to  the Hamiltonian $\omega N$, so it is also 
 invariant under the  unitary evolution generated by  $\K_N$: For 
 any  $ t\in \R$, we have 
\begin{equation}
     e^{-it\omega\K_N}(\rho_{\text{ph}, \text{G}}) = \rho_{\text{ph}, \text{G}}  \,.   
\label{Gibbsinv}
\end{equation}

\section{Infinite-component vector spin system}  
\label{sec:SPIN}
For the matter, we consider  an  infinite-component spin system \cite{D1}.  We introduce this system  by employing 
 the algebraic  formulation  as  in \S2.

Let $\G$ be a Hilbert space and $\{ \, |n) \}_{n\in \Z}$
be  a complete 
orthonormal system  of $\G$. 
As operators on $\G$, we consider $l_{\pm}$ and $M$ defined  by
\begin{equation}
  M\, |n) = n\, |n), \quad
           l_{\pm}|n) = \, |n \pm 1) 
          \qquad ( n \in \Z ) \,.   
\label{Mspec}
\end{equation}
The following fundamental relations hold: 
\begin{equation}
       [M, l_{\pm}] = \pm l_{\pm}, \qquad 
       l_+l_-=l_-l_+ = \bold{1}.
\label{com1}
\end{equation}

\begin{rem} 
\label{rem:ALG}    
\rm{
In the Dicke model, the $n$-compound system of two-component spins is resolved into a set of irreducible representations. 
Picking a representation whose total spin is $\ell$ from the set, we 
realize it in $\G$ as 
\[
J_{\pm,\ell} \, |m) = \sqrt{(\ell \mp m)(\ell \pm m + 1)} \, |m \pm 1), \quad 
J_{3, \ell} \, |m) = m \, |m) \qquad ( \, |m| \leqslant \ell \,) ,
\]
\[
     J_{\pm,\ell} \, |m) = J_{3, \ell} \, |m) = 0 \qquad ( \, |m| > \ell \,)\,. 
\]
Then we have operators $l_{\pm}$ as  
\[
    {\textrm{s-}}\!\lim_{\ell \to \infty} \frac{J_{\pm,\ell}}{\ell} 
    = l_{\pm}, 
\]
where s-lim stands for the strong limit of operators on the Hilbert space $\G$, that is, 
$ J_{\pm,\ell} |\psi\rangle /\ell \to l_{\pm} |\psi \rangle $ in 
 the norm of $\G$ for arbitrary $ |\psi \rangle \in \G$.
The convergence $ J_{3,\ell} \to M$ also holds in norm on the domain of
$M$. Moreover in \cite{D1}, it has been shown that the Hamiltonian  
\begin{equation}
      H_{\ell} = \omega \, \bold{1}\otimes a^{\dagger}a 
     + \mu \, J_{3, \ell}\otimes\bold{1} + \frac{\lambda}{\ell} 
     \big(J_{+, \ell}\otimes a +  J_{-, \ell}\otimes a^{\dagger}\big)
\label{lSDHamiltonian}
\end{equation}
of the $(2\ell +1)$-component Dicke model acting in $\G\bigotimes\F$ 
converges to the Hamiltonian 
\begin{equation}
     H = \omega \, \bold{1}\otimes a^{\dagger}a 
    + \mu \, M\otimes\bold{1}   
       + \lambda\big(l_+\otimes a +  l_-\otimes a^{\dagger}\big) 
\label{ISDHamiltonian}
\end{equation}
of the ISD model in the strong generalized sense, i.e., 
\[
      \textrm{s-}\!\lim_{\ell\to\infty}e^{-itH_{\ell}} = 
      e^{-itH} 
\]
on $\G\bigotimes\F$.  
Intuitively, those strong limits give the idealization of the property 
of corresponding operators that hold for the action only on the 
vectors consisting of linear combinations of $|m)$ satisfying 
$ |m| \ll \ell$ for large but finite $\ell$.   
In the rest of this section, we deal with the dissipative infinite-component 
spin system described by $l_{\pm}$ and $M$.    
We consider the open system on $\G\bigotimes\F$ given by the 
Hamiltonian (\ref{ISDHamiltonian}) and dissipation terms in the next section.   }
\end{rem}

We consider the master equation
\begin{equation}
     \frac{d}{dt}\rho(t) = \L_{\text{sp}}(\rho(t))  \qquad \mbox{for}
     \qquad \rho(t) \in \T(\G) \,,
\label{master1}
\end{equation}
where the   Liouvillian is defined by 
\begin{equation}
   \L_{\text{sp}} :=  -i\mu\K_{M} 
  + \alpha_-\D_{l_-\circ l_+} + \alpha_+ 
  \D_{l_+\circ l_-}  
\label{LIOU-spin}
\end{equation}
with constants $\mu > 0, \, \alpha_{\pm} \geqslant 0$.   
The subscript of $\L_{\text{sp}}$ signifies ``spin.''   
We consider  mainly the case of $\alpha_{\pm}=0$,  
i.e., no dissipation for the spin as a component of the OISD model in \S\ref{sec:OISD}.  
However, it will become  clear that 
 treating both  cases with or without dissipation 
 on an equal footing
 is   helpful for our discussion.

We discuss the   behavior of the open quantum 
 dynamics  generated by the Liouvillian $\L_{\text{sp}}$.
Its qualitative picture may be described as follows.
Because its Hamiltonian part generated by $M$ is not lower bounded (\ref{Mspec}), 
  it is considered  to be unstable.   
One may guess that 
if  we add  dissipation   
as  in the Liouvillian
(\ref{LIOU-spin}), then  any  density matrix made by
 the  eigenstates  $\{ \, |n) \}_{n\in \Z}$ for  $M$ will be  no
 longer  stable under   the dissipative dynamics
  generated by  $\L_{\text{sp}}$.  
Actually, we verify  this  naive  picture in the following.  

We consider the algebraic structure of the Liouvillian  $\L_{\text{sp}}$ and its 
constituents.    
From (\ref{KK}), (\ref{DD}), and (\ref{com1}), we have
\begin{equation}
    [\K_{M}, \D_{l_-\circ l_+}] = 
    [\K_{M}, \D_{l_+\circ l_-}] =
    [\D_{l_-\circ l_+}, \D_{l_+\circ l_-}] = 0 \,.
\label{comKDD1} 
\end{equation}
We  have also 
\begin{equation}
  e^{-i\mu t\K_M}(\rho) = e^{-i\mu tM}\rho 
     \, e^{i\mu tM}\,, 
\end{equation}
\begin{equation}
   e^{t\alpha_-\D_{l_-\circ l_+}}(\rho) =
   \sum_{n=0}^{\infty}\frac{(2t\alpha_-)^n}{n!}
   e^{-2t\alpha_-} l_-^n\rho l_+^n,
\end{equation}
and
\begin{equation}
   e^{t\alpha_+\D_{l_+\circ l_-}}(\rho) =
   \sum_{m=0}^{\infty}\frac{(2t\alpha_+)^m}{m!}
   e^{-2t\alpha_+} l_+^m\rho l_-^m   
\end{equation}
for  $ \rho \in \T(\G)$.  
From these expressions combined with the commutativity relations 
(\ref{comKDD1}), 
the solution of the master equation  is given as
\begin{align}
   e^{t\L_{\text{sp}}}(\rho) =& e^{t\alpha_-\D_{l_-\circ l_+}}
   e^{t\alpha_+\D_{l_+\circ l_-}}
   e^{-i\mu t\K_M}(\rho)
\notag \\
   = & \sum_{n,m =0}^{\infty}
    \frac{(2t\alpha_-)^n(2t\alpha_+)^m}{n!m!}
    e^{-2t(\alpha_-+\alpha_+)}
   l_-^{n-m} e^{-i\mu tM}\rho 
     \, e^{i\mu tM}  l_+^{n-m}
 \notag \\
   = & \sum_{k= -\infty}^{\infty}c_k(t)
   l_-^k e^{-i\mu tM}\rho 
     \, e^{i\mu tM}  l_+^k,
\label{evol1}
\end{align}
where
\begin{equation}
     c_k(t) \equiv \sum_{n,m =0}^{\infty}\delta_{n-m, k}
     \frac{(2t\alpha_-)^n(2t\alpha_+)^m}{n!m!}
    e^{-2t(\alpha_-+\alpha_+)} \,.
\label{ckt}
\end{equation}

The set of coefficients $\{ \, c_k(t) \, \}_{k \in \Z}$ 
given in (\ref{ckt}) may be 
considered as a time-dependent probability distribution in the sense 
\begin{equation}
           c_k(t) \geqslant 0, \qquad 
          \sum_{k= -\infty}^{\infty}c_k(t) = 1 \, ,
\label{prob}
\end{equation}
from which it follows that the evolution (\ref{evol1}) has the CPTP property.   
On the other hand, the behavior of its mean and variance, namely 
\begin{equation}
    \sum_{k= -\infty}^{\infty} k c_k(t) = 
      2(\alpha_--\alpha_+)t, \qquad 
      \sum_{k= -\infty}^{\infty} k^2 c_k(t)
    - \Big(\sum_{k= -\infty}^{\infty} k c_k(t) \Big)^2 
       = 2(\alpha_-+\alpha_+)t,  
\label{meanvariance}
\end{equation}
exhibits the floating and diffusive nature of the evolution.  
Moreover, it can be shown that 
\begin{equation}
    \textrm{s-}\!\lim_{t \to \infty}e^{t\L_{\text{sp}}}(\rho) = 0
\label{slim}
\end{equation}
holds for arbitrary initial state $\rho$ 
unless $(\alpha_{+}, \alpha_{-} ) =  (0,0) $.     
In particular, with such  non-trivial dissipation, 
there exists  no eigenstate for 
$\L_{\text{sp}}$ and there exists  no  steady 
(i.e., temporally invariant) state for the dynamical semigroup 
$\{e^{t\L_{\text{sp}}}\}_{t \geqslant 0}$.     
As an approximation to the $(2\ell + 1)$-component spin system, 
the evolution formula (\ref{evol1}) should not be used beyond the 
restriction $t\ll \ell/2(\alpha_+ + \alpha_-)$.     
\section{{Open infinite-spin Dicke model}}  
\label{sec:OISD}
In this section, 
 we    investigate the OISD model,  
 an open quantum model  generated 
by  the infinite-component  spin in  \S\ref{sec:SPIN} 
and  the open harmonic oscillator  in \S\ref{sec:HO} 
with the Jaynes--Cummings interaction between them.  

We provide the  precise formulation of the OISD model in the following.   
The Hilbert space of the  system is $\H = \G \bigotimes \F$. 
Operators acting in $\H$ such as  
$M\otimes \bold{1}$,  $\bold{1}\otimes a^{\dagger}a$, and $l_+ \otimes a$ 
 are denoted simply as $M$, $a^{\dagger}a$, and $l_+a$, respectively, by  obvious embedding.   
We introduce  the shorthand  notation
\begin{align}
   \K_{\text{sp}} \equiv  \K_M, \quad  \K^{\text{int}}_{-+} 
\equiv  \K_{l_- a^{\dagger}}, 
   \quad \K^{\text{int}}_{+-} \equiv  \K_{l_+a}, \quad 
   \K^{\text{int}} \equiv  \K^{\text{int}}_{+-}+\K^{\text{int}}_{-+}, 
\quad 
\K_{\text{ph}}\equiv  \K_{N}, 
\end{align}
and for nonnegative constant $J$, 
\begin{align}
\D_{\text{sp}} \equiv  & (J+1)\D_{l_-\circ l_+} + J\D_{l_+\circ l_-},\  
\D_{\text{ph}} \equiv  (J+1)\D_{a\circ a^{\dagger}} + J\D_{a^{\dagger}\circ a},   \notag  \\
\D^{\text{int}}_{-+}   \equiv  & (J+1)\D_{l_-\circ a^{\dagger}} + J\D_{a^{\dagger}\circ l_-},\  \D^{\text{int}}_{+-} \equiv   (J+1)\D_{a\circ l_+} + J\D_{l_+\circ a}.
\end{align}
With this  notation,  the Liouvillian of the OISD model is defined  by 
\begin{equation}
   \L_{\text{OISD}}
 :=  -i\mu\K_{\text{sp}}  -i\lambda \K^{\text{int}} -i\omega\K_{\text{ph}} 
       + \gamma\D_{\text{ph}}   
\label{LIOU-OISD}
\end{equation}
with positive constants $\omega, \mu$, and $\gamma$.   
As we have anticipated, 
 $\K^{\text{int} } $ is  
the Jaynes--Cummings interaction between  the spin and the  oscillator, and  
the constant  $\lambda\in \R$ denotes the strength of this interaction.   
The time evolution on the composed system 
is governed by the above Liouvillian as  
\begin{equation}
     \frac{d}{dt}\rho(t) = \L_{\text{OISD}}(\rho(t)) \qquad \mbox{for}\quad \rho(t) \in \T(\H) \,.
\label{master}
\end{equation}  

Note that  in  the Liouvillian  
(\ref{LIOU-OISD}) of the OISD model,  
  dissipation is induced  
 only through  the  harmonic oscillator (not through the spin).
Later,  we discuss   a  more general   Liouvillian 
that  has   dissipation terms   both for  
the  harmonic oscillator and the infinite-component spin. 
It turns out that the analysis for such a general  model 
is essentially reduced to the 
 simple case  (\ref{LIOU-OISD}). Hence we focus 
 on this special setup for our OISD model.

Let us analyze the algebraic structure of the Liouvillian 
$\L_{\text{OISD}}$ in (\ref{LIOU-OISD}).
 We can straightforwardly check the following commutation relations:
\begin{align}
&[\K_{\text{ph}}, \K^{\text{int}}_{-+}] 
= - [\K_{\text{sp}}, \K^{\text{int}}_{-+}] 
   =  [\D_{\text{ph}}, \D^{\text{int}}_{-+}] = + \K^{\text{int}}_{\mp\pm}, \notag\\
  &[\K_{\text{ph}}, \K^{\text{int}}_{+-}] 
= - [\K_{\text{sp}}, \K^{\text{int}}_{+-}] 
   =  [\D_{\text{ph}}, \D^{\text{int}}_{+-}] = - \K^{\text{int}}_{\mp\pm}, \notag\\
 & [\D_{\text{ph}}, \K^{\text{int}}_{-+}] = + \D^{\text{int}}_{-+},\ 
  [\D_{\text{ph}}, \K^{\text{int}}_{+-}] = - \D^{\text{int}}_{+-} , \notag\\ 
&  [\K^{\text{int}}_{-+}, \D^{\text{int}}_{+-}] = - \D_{\text{sp}} ,\  
  [\K^{\text{int}}_{+-}, \D^{\text{int}}_{-+}] = + \D_{\text{sp}} , 
\label{KDpm} 
\end{align}
and 
\begin{align}
   [\K_{\text{ph}}, \D_{\text{ph}}] = [\K_{\text{sp}}, \D_{\text{ph}}] 
   = [\K^{\text{int}}_{-+}, \K^{\text{int}}_{+-}] = 
  [\K^{\text{int}}_{-+}, \D^{\text{int}}_{-+}] 
 =  [\K^{\text{int}}_{+-}, \D^{\text{int}}_{+-}]=0, \notag\\ 
  [\D_{\text{sp}}, \K^{\text{int}}_{-+}] 
 = [\D_{\text{sp}}, \K^{\text{int}}_{+-}] 
   = [\D_{\text{sp}}, \D_{\text{ph}}] = 0.
\label{KD0}
\end{align}

For $\eta \in \C$, define 
\begin{equation}
\label{DEFWeta}                 
W(\eta) := e^{\eta \K^{\text{int}}_{-+} - \bar{\eta}\K^{\text{int}}_{+-}}.    
\end{equation}
We see  that for any $\eta \in \C$,  
$W(\eta)$ is  a  CPTP map with  
 its  bounded inverse $W(-\eta)$. This fact will be important later.  
From (\ref{KDpm}) and (\ref{KD0}), we have 
\begin{equation}
W(\eta) \left(
        \begin{array}{c}
        \K_{\text{sp}} \\
        \K^{\text{int}}_{-+} \\
        \K^{\text{int}}_{+-} \\
         \K_{\text{ph}} \\
        \D_{\text{sp}} \\
        \D^{\text{int}}_{-+}  \\
        \D^{\text{int}}_{+-}  \\
        \D_{\text{ph}} 
        \end{array}
    \right) W(-\eta) = \left(\begin{array}{c}
        \K_{\text{sp}} + \eta \K^{\text{int}}_{-+} + \bar{\eta}\K^{\text{int}}_{+-} \\
  \K^{\text{int}}_{-+} \\
        \K^{\text{int}}_{+-} \\
        \K_{\text{ph}} - \eta \K^{\text{int}}_{-+} - \bar{\eta}\K^{\text{int}}_{+-} \\
        \D_{\text{sp}} \\
        \D^{\text{int}}_{-+} - \bar{\eta}\D_{\text{sp}}  \\
        \D^{\text{int}}_{+-} -\eta\D_{\text{sp}} \\
        \D_{\text{ph}} - \eta \D^{\text{int}}_{-+} - \bar{\eta}\D^{\text{int}}_{+-}  
   + |\eta|^2\D_{\text{sp}} 
        \end{array}
    \right),
\label{vecCR1}
\end{equation}
where the adjoint  action $\text{Ad}(W(\eta))$ acts on each component.   
Note that higher terms of $\eta$ and $ \bar{\eta} $ 
 vanish in the right-hand side because of commutativity (\ref{KD0}).    
For any $ \sigma, t \geqslant 0$, 
\begin{equation}
           e^{\sigma\D_{\text{ph}}}\left(
      \begin{array}{c}
      \K^{\text{int}}_{-+} \\ \K^{\text{int}}_{+-} \\ \D^{\text{int}}_{-+} \\ 
      \D^{\text{int}}_{+-}
      \end{array} \right) = \left(
      \begin{array}{c}
      \K^{\text{int}}_{-+} \cosh \sigma + \D^{\text{int}}_{-+} \sinh \sigma \\
      \K^{\text{int}}_{+-}\cosh \sigma  - \D^{\text{int}}_{+-} \sinh \sigma \\ 
      \D^{\text{int}}_{-+}\cosh \sigma  + \K^{\text{int}}_{-+} \sinh \sigma \\ 
     \D^{\text{int}}_{+-}\cosh \sigma  - \K^{\text{int}}_{+-} \sinh \sigma
      \end{array} \right) e^{\sigma\D_{\text{ph}}}
\label{vecCR2}                 
\end{equation}
and
\begin{equation}
           e^{-it(\mu\K_{\text{sp}} +\omega\K_{\text{ph}})}\left(
      \begin{array}{c}
      \K^{\text{int}}_{-+} \\ \K^{\text{int}}_{+-}      
      \end{array} \right) = \left(
      \begin{array}{c}
      e^{-it(\omega - \mu)}\K^{\text{int}}_{-+} \\
      e^{it(\omega - \mu)}\K^{\text{int}}_{+-}
      \end{array} \right) e^{-it(\mu\K_{\text{sp}} +\omega\K_{\text{ph}})}  
\label{vecCR3}                 
\end{equation}
hold.   
These identities can be  verified by repeating an argument  similar to that used in  
(\ref{transform0}).

We can see that  $e^{t\L_{\text{OISD}}}$ is a CPTP map on $\T(\H)$   
for each $t \geqslant 0 $ as follows.
 Note that 
 $e^{t\gamma \D_{\text{ph}}}$ is identical to  
      ${\bf{1}}\otimes \exp\big(t\gamma (J+1)\D_{a\circ a^{\dagger}} 
      + t\gamma J\D_{a^{\dagger}\circ a} \big)$.
Because its  second factor  is completely positive on 
$\T(\F)$ as we have seen in \S\ref{sec:HO}, 
  $e^{t\gamma \D_{\text{ph}}}$ is   
a completely positive map  on $ \T(\H) = \T(\G) \bigotimes \T(\F)$.     
The trace-preserving property  follows 
from the tensor-product structure as well.  
Similarly, $e^{t\gamma'\D_{\text{sp}}}$ and thereby 
$e^{t(\gamma\D_{\text{ph}} + \gamma'\D_{\text{sp}})} = 
e^{t\gamma\D_{\text{ph}}}e^{t\gamma'\D_{\text{sp}}} $ 
are CPTP maps for $\gamma' > 0$.  
We have the following  formula 
 of the semigroup generated by $\L_{\text{OISD}}$: 
\begin{equation}
   e^{t\L_{\text{OISD}}} = e^{-it(\mu\K_{\text{sp}} +\omega\K_{\text{ph}})}
            W\big(\eta_1(t)\big)e^{t\gamma\D_{\text{ph}} +\tau_2(t)\D_{\text{sp}}}
            W\big(\eta_2(t)\big) \,,  
\label{prod_L}
\end{equation}
where $\eta_1(t), \eta_2(t)$, and $\tau_2(t)$ are the solutions of the differential equations
\begin{align}
   &\eta_2'(t) \sinh t\gamma - \gamma\eta_1(t) = 0, 
\label{eta2} \\
  &\big(\eta_1'(t) + \gamma \eta_1(t)\coth t\gamma\big)e^{-it(\omega-\mu)} +i\lambda = 0, 
\label{eta1} \\
  &\tau_2'(t) - \gamma |\eta_1(t)|^2 = 0
\label{tau2}
\end{align}
with initial condition $ \eta_1(0) = \eta_2(0) = \tau_2(0) = 0$.  
The solution for $\eta_1(t)$ is given explicitly by 
\begin{equation}
         \eta_1(t) = \frac{i\lambda}{(\omega-\mu)^2 + \gamma^2}
          \Big( i(\omega - \mu)e^{i(\omega-\mu)t} 
          + \frac{\gamma(1-e^{i(\omega-\mu)t}\cosh \gamma t)}{\sinh \gamma t} \Big) \,, 
\end{equation}
and $\eta_2(t)$ and $\tau_2(t)$ are obtained readily from (\ref{eta2}) and (\ref{tau2}).  
The identity  (\ref{prod_L})  can be shown as in  (\ref{evolutionD}).   
Namely  
by putting the right hand-side by $Z(t)$, we have
\begin{equation}
          Z(0) = {\bf{1}}, \qquad \frac{dZ(t)}{dt} = \L_{\text{OISD}} (Z(t))
\end{equation}
with the help of (\ref{vecCR1}), (\ref{vecCR2}), and (\ref{vecCR3}).   
Because each factor in the right-hand side of (\ref{prod_L}) is CPTP, $e^{t\L_{\text{OISD}}}$ is also CPTP.

Next  we derive  the decomposition formula 
of the Liouvillian $\L_{\text{OISD}}$. For this purpose, we need a  
 map that gives rise to the decomposition.  
 By using $\D_{\text{ph}}$ and 
$W(\eta)$ defined  in (\ref{DEFWeta}),                 
we set 
\begin{equation}
   V(\sigma)  := W(\zeta_1) e^{\sigma\D_{\text{ph}}} W(\zeta_2), \qquad 
\sigma >0,  
\label{V-DEF}
\end{equation}
where  the constants are defined  by
\begin{equation}
  \delta := \frac{\lambda}{\gamma^2 + (\omega - \mu)^2} \,,  
  \quad 
   \zeta_1 := -(\omega - \mu + i\gamma\coth \sigma)\delta \,, \quad 
   \zeta_2 := \frac{i\gamma\delta}{\sinh \sigma}. 
\label{V-constans-function}
\end{equation}
From (\ref{vecCR1}) and (\ref{vecCR2}),
we obtain
\begin{align*}
\Bigl(  -i\mu\K_{\text{sp}}   -i\lambda \K^{\text{int}}  
-i\omega\K_{\text{ph}} 
+ \gamma\D_{\text{ph}}\Bigr) V(\sigma)
    = V(\sigma) & \Bigl(  - i\mu\K_{\text{sp}}    
    +  \lambda \gamma\delta\D_{\text{sp}}  
    -i\omega\K_{\text{ph}}  + \gamma\D_{\text{ph}}  \Bigr).
\end{align*}
Thus, we arrive at our  first main result as follows.  

\begin{prop}
\label{prop:decom}
The Liouvillian  $\L_{\text{OISD}}$ of the OISD model 
 is decomposed into the  decoupled 
Liouvillian $\L_\text{decoupled}$
 under the similarity transformation induced by $V(\sigma)$ 
 as 
\begin{align}
      &\L_{\rm{OISD}} V(\sigma)=V(\sigma)\L_{\rm{decoupled}},\notag\\
&\L_{\rm{decoupled}} := \tilde\L_{\rm{sp}} + \tilde\L_{\rm{ph}},\quad 
\tilde\L_{\rm{sp}}:=- i\mu\K_{\rm{sp}}
 +  \lambda \gamma\delta\D_{\rm{sp}},\quad  
 \tilde\L_{\rm{ph}}:=-i\omega\K_{\rm{ph}}  + \gamma\D_{\rm{ph}}. 
\label{SPLIT} 
\end{align}
The transformation  $V(\sigma)$  used above 
 is an invertible  CPTP map  on  $\T(\H)$. 
\end{prop}
Note that 
 $\tilde\L_{\text{ph}}$ is   
 the Liouvillian (\ref{LIOU-ho}) acting in  $\T(\F)$ 
(imbedded into $\T(\H)$), and 
$\tilde\L_{\text{sp}}$ is  the Liouvillian
 (\ref{LIOU-spin}) acting in  $\T(\G)$ 
(imbedded into $\T(\H)$) 
with values  $\alpha_- 
= \lambda\gamma\delta (J+1)$, $\alpha_+ =  \lambda\gamma\delta J$.
The  formula  (\ref{SPLIT}) tells us
 that the similarity transformation by $V(\sigma)$ 
erases the Jaynes--Cummings interaction $-i\lambda \K^{\text{int}}$ in the original  $\L_{\text{OISD}}$  but generates  the  new 
dissipation term $\lambda \gamma\delta\D_{\text{sp}}$ in 
the new decoupled Liouvillian $\L_{\rm{decoupled}}$.

 To complete the proof of Proposition~\ref{prop:decom}, it remains 
 to show that 
 $V(\sigma)$ is an invertible CPTP map.
Note that $W(\zeta_1)$ and $W(\zeta_2)$
are CPTP maps with  
 their bounded inverses $W(-\zeta_1)$ and $W(-\zeta_2)$, respectively. 

Furthermore, note that  $e^{\sigma\D_{\text{ph}}}$ is  
a CPTP map,  and that it has its   dense range   
 and (unbounded) inverse 
because of the spectrum  of $\D_{\text{ph}}$
given in (\ref{JD0eigen}).   Thus  $V(\sigma)$ is an invertible CPTP 
map on $\T(\H)$ with its dense range.

\bigskip 

In the following, we discuss  the asymptotic behavior of 
 the dynamical semigroup  $\{ e^{t \L_{\text{OISD}}} \}_{t\geqslant 0}$.
To this end, let us rewrite 
the  decomposition formula  (\ref{SPLIT})
 in tensor-product format  as
\begin{equation}
      \L_{\text{OISD}} V(\sigma)=V(\sigma)\L_\text{decoupled},\quad \L_\text{decoupled} := \tilde\L_{\text{sp}}\otimes \bold{1} + \bold{1}\otimes \tilde\L_{\text{ph}} \,.
\label{LIOU-Trans-tensor}
\end{equation}
The solution for  the master equation (\ref{master}) of the OISD model 
is written formally as 
\begin{equation}
  \rho(t) \equiv e^{t\L_{\text{OISD}}} \big(\rho(0)\big) 
 = V(\sigma)\Big( e^{t\tilde\L_{\text{sp}}} 
\otimes e^{t\tilde\L_{\text{ph}}}\Big)V(\sigma)^{-1} \big(\rho(0)\big)   \,. 
\label{PRODUCT-evolution}
\end{equation}
By applying an argument similar to that leading to (\ref{toGibbs})   
at the end of \S \ref{sec:HO} to the 
second tensor component of 
 $\rho(t)$ in the right-hand side of 
(\ref{PRODUCT-evolution}), we obtain  the convergence  
\begin{equation}
         \Big(\bold{1}\otimes e^{t\tilde\L_{\text{ph}}}\Big)V(\sigma)^{-1}
     \big(\rho(0)\big) \rightarrow \rho_\ast\otimes\rho_{\text{ph}, \text{G}}    \label{toastGIB}
\end{equation}
\bigskip
in norm as $ t \to \infty$, where $\rho_\ast $ is some uniquely determined 
 density  matrix in $\T(\G)$ 
and $\rho_{\text{ph}, \text{G}} \in \T(\F)$ is  the Gibbs state
 as  in (\ref{Gibbs0}). 

Let us modify the decoupled Liouvillian  $\L_{\rm{decoupled}} = - i\mu\K_{\rm{sp}} +  \lambda \gamma\delta\D_{\rm{sp}}
-i\omega\K_{\rm{ph}}  + \gamma\D_{\rm{ph}}$ of  (\ref{SPLIT})
by setting  $\omega=\mu$ and deleting  the term $\gamma \D_{\rm{ph}}$ as
\begin{equation}
         \check\L_{\rm{syn.dec.}} :=
-i\mu(\K_{\text{sp}}+\K_{\text{ph}}) + \lambda\gamma\delta
\D_{\text{sp}}.
 \label{SAME}
\end{equation}
Let us call it  {\it{the synchronized decoupled Liouvillian}}.
Note that it differs from $\L_{\rm{decoupled}}$  with  
$\omega = \mu $ (and  $\gamma=0$).   
We show the following proposition. 

\begin{prop}
\label{prop:AYSMP}
Let $\rho(0)$ be any density matrix of  $\T(\H)$, 
 and let $\rho(t)$ $(t \geqslant 0)$ denote its time evolution 
under the OISD Liouvillian $\L_{\rm{OISD}}$.
Define the density matrix   
\begin{align}
\check\rho(0) :=
 V(\sigma)\big(\rho_\ast\otimes\rho_{\rm{ph}, \rm{G}} \big)\in \T(\F)
\label{newINITIAL}
\end{align}
and consider its time evolution  
under  the  synchronized decoupled Liouvillian
 $\check\L_{\rm{syn.dec.}}$ given as  
\begin{align}
  \check\rho(t):= 
 e^{t\check\L_{\rm{syn.dec.}}
} \big(\check\rho(0) \big).
\label{SYNCRONIZED}
\end{align}
Then 
\begin{align}
\lim_{t\to \infty} \| \rho(t) - \check\rho(t) \|_1=0.
\label{ASYM}
\end{align}
Namely, the time evolution of any state $\rho(0)$ under the OISD  Liouvillian 
 is asymptotically equivalent  
to the time evolution  of the state 
$\check\rho(0)$ under the synchronized decoupled Liouvillian.    
If  $\lambda\gamma\delta\ne 0$, then 
\begin{align}
  \textrm{s-}\!\lim_{t\to \infty}  \rho(t)=0.
\label{COMPLETE-DISSIPATED}
\end{align}
\end{prop}

\begin{proof}

We first note the following remarkable fact.
 For  $\check\rho(0)$, which has the  specific form as defined in 
 (\ref{newINITIAL}), 
its time evolution under  $\check\L_{\rm{syn.dec.}}$ 
is identical to the  time evolution under  
the genuine OISD Liovillian $\L_{\text{OISD}}$, 
 namely for all  $t\geqslant 0$, we have
\begin{align}
\check\rho(t)\equiv 
 e^{t\check\L_{\rm{syn.dec.}}} \big(\check\rho(0) \big)
=  e^{t\L_{\text{OISD}}}\left(\check\rho(0)  \right).
\label{TWO-LIUS-ONAJI}
\end{align}
We show this identity  in the following. We compute  the right-hand side 
 of (\ref{TWO-LIUS-ONAJI}) as
\begin{align}
 e^{t\L_{\text{OISD}}}\left(\check\rho(0)  \right)&=
e^{t\L_{\text{OISD}}}\left(V(\sigma)\big(\rho_\ast\otimes\rho_{\text{ph}, \text{G}} \big)\right)\notag\\ 
   &=V(\sigma) e^{t\L_\text{decoupled}} \big(\rho_\ast\otimes\rho_{\text{ph}, \text{G}} \big) \label{KOUKANV} \\
 &= V(\sigma)\Big( e^{t\tilde\L_{\text{sp}}}\otimes \bold{1}\Big) 
\Big(\bold{1}\otimes e^{t\tilde\L_{\text{ph}}}\Big) \big(\rho_\ast\otimes\rho_{\text{ph}, \text{G}} \big)  \notag    
\\ & = V(\sigma)\Big( e^{t\tilde\L_{\text{sp}}}\otimes \bold{1}\Big)
 \big(\rho_\ast\otimes\rho_{\text{ph}, \text{G}} \big).
\label{Vatamadashi}
\end{align}
For the derivation of (\ref{KOUKANV})
the relation (\ref{LIOU-Trans-tensor}) is used, and 
for (\ref{Vatamadashi})
the invariance  of the Gibbs state $\rho_{\text{ph}, \text{G}}$
 under $e^{t\tilde\L_{\text{ph}}}$ given as (\ref{INVGibbs}) is used.
Substituting the  identity   $e^{-it\mu\K_{\text{ph}}} 
(\rho_{\text{ph}, \text{G}})=\rho_{\text{ph}, \text{G}}$
(which  obviously holds as in   (\ref{Gibbsinv})) into  (\ref{Vatamadashi}),  
we obtain  
\begin{align}
   e^{t\L_{\text{OISD}}}\left(\check\rho(0)  \right)&=
 V(\sigma)\Big( e^{t\tilde\L_{\text{sp}}}\otimes \bold{1}\Big) 
\Big(\bold{1}\otimes e^{-it\mu\K_{\text{ph}}}\Big) \big(\rho_\ast\otimes\rho_{\text{ph}, \text{G}} \big)     \notag
\\ & = V(\sigma) e^{t(-i\mu\K_{\text{sp}}+ \lambda\gamma\delta\D_{\text{sp}} 
-i\mu\K_{\text{ph}})}  
\big(\rho_\ast\otimes\rho_{\text{ph}, \text{G}} \big)      \notag
\\ & = V(\sigma) e^{t(-i\mu(\K_{\text{sp}}+\K_{\text{ph}} )+ \lambda\gamma\delta\D_{\text{sp}}) }  
\big(\rho_\ast\otimes\rho_{\text{ph}, \text{G}} \big)     
 \label{VmuomePROD}
\\ & =  e^{t(-i\mu(\K_{\text{sp}}+\K_{\text{ph}}) + \lambda\gamma\delta
\D_{\text{sp}})} 
 V(\sigma)\big(\rho_\ast\otimes\rho_{\text{ph}, \text{G}} \big) 
 \label{muome}
\\ & 
= e^{t\check\L_{\rm{syn.dec.}}[\mu=\omega]}
 V(\sigma)  \big(\rho_\ast\otimes\rho_{\text{ph}, \text{G}} \big)  
=
 e^{t\check\L_{\rm{syn.dec.}}}
\left(\check\rho(0)  \right),
\notag
\end{align}
 which gives  our desired (\ref{TWO-LIUS-ONAJI}).
To show   (\ref{muome}),
  we note   the  commutativity relations  
$
[\D_{\text{sp}}, \K^{\text{int}}_{-+}]
 =
[\D_{\text{sp}}, \K^{\text{int}}_{+-}]
   = [\D_{\text{sp}}, \D_{\text{ph}}] =  
[\K_{\text{sp}} + \K_{\text{ph}}, \D_{\text{ph}}] =
[\K_{\text{sp}} + \K_{\text{ph}}, \K_{-+}^{\text{int}}]
=
[\K_{\text{sp}} + \K_{\text{ph}}, \K_{+-}^{\text{int}}]
 = 0$,
 which  follow from (\ref{KDpm}) and (\ref{KD0}),  and then 
  recall  the definition of $V(\sigma)$
 given in 
 (\ref{DEFWeta}) and (\ref{V-DEF}). All the   ingredients $\K^{\text{int}}_{-+}$, 
 $\K^{\text{int}}_{+-}$, and $\D_{\text{ph}}$ 
 generating $V(\sigma)$ 
commute with each of  
$\D_{\text{sp}}$ and  $\K_{\text{sp}} + \K_{\text{ph}}$.

Noting (\ref{PRODUCT-evolution}) and (\ref{Vatamadashi}),
we have the estimate 
\begin{align}
& \| \rho(t) - \check\rho(t) \|_1  \notag
\\ 
 = & \left\| V(\sigma)\Big( e^{t\tilde\L_{\text{sp}}}\otimes 
\bold{1}\Big)\Big(\bold{1}\otimes e^{t\tilde\L_{\text{ph}}}\Big)
V(\sigma)^{-1} \big(\rho(0)\big) - V(\sigma)\Big( e^{t\tilde\L_{\text{sp}}}
\otimes \bold{1}\Big) 
 \big(\rho_\ast\otimes\rho_{\text{ph}, \text{G}} \big) \right\|_1 \notag
\\
 = & \left\| V(\sigma)\Big( e^{t\tilde\L_{\text{sp}}}\otimes 
\bold{1}\Big)\Big[\Big(\bold{1}\otimes e^{t\tilde\L_{\text{ph}}}\Big)
V(\sigma)^{-1} \big(\rho(0)\big) - \rho_\ast\otimes\rho_{\text{ph}, \text{G}} \Big] \right\|_1 \notag
\\
 \leqslant & \left \| \Big[\Big(\bold{1}\otimes e^{t\tilde\L_{\text{ph}}}\Big)
V(\sigma)^{-1} \big(\rho(0)\big) - \rho_\ast\otimes\rho_{\text{ph}, \text{G}} \Big] 
\right \|_1. 
\label{EST}
\end{align}
The last inequality is a consequence of the fact that CPTP maps act 
 in a contractive manner  
on self-adjoint elements of $\T(\H)$.    
By combining (\ref{toastGIB})
 with  (\ref{EST}), we have (\ref{ASYM}).

We show (\ref{COMPLETE-DISSIPATED}).
By (\ref{TWO-LIUS-ONAJI}) and (\ref{Vatamadashi}), 
we have 
\begin{align}
\check\rho(t)
 = V(\sigma)
 \left(e^{t\tilde\L_{\text{sp}}}(\rho_\ast)
\otimes\rho_{\text{ph}, \text{G}} \right).
\label{CHECKRHO-dec}
\end{align}
If  $\lambda\gamma\delta\ne 0$, equivalently, 
 the dissipation part of 
$\tilde\L_{\rm{sp}}=- i\mu\K_{\rm{sp}}
 +  \lambda \gamma\delta\D_{\rm{sp}}
 $ is non-zero, 
then 
according to the argument in the last part of   \S\ref{sec:SPIN}, we have  
\begin{align*}
    \textrm{s-}\!\lim_{t\to \infty}  e^{t\tilde\L_{\text{sp}}}(\rho_\ast) =0, 
\end{align*}
therefore   
\begin{align}
    \textrm{s-}\!\lim_{t\to \infty}  \check\rho(t) = 0
\label{check-DISSIPATED}
\end{align}
holds.   
Because the trace  norm convergence implies the strong  convergence, we obtain 
  $\textrm{s-}\!\lim_{t\to \infty} \rho(t)=0$  
from (\ref{ASYM}) and  (\ref{check-DISSIPATED}).  
\end{proof}

\begin{rem} 
\label{rem:GCZ}    
{\rm{
From the above argument, we see that 
the OISD model demonstrates  a typical  mechanism of 
synchronization, namely, the separation of multiple dissipative time scales 
 as described  in \cite{GCZ}.   
The essential point lies in  the 
different decay rates between the two systems:
 the dressed photon decays  into 
the Gibbs state with the positive decay rate  $\gamma$, 
  while the dressed spin does not decay into a certain state in a finite 
time scale.}}    
\end{rem}

\begin{rem} 
\label{rem:BULL}    
{\rm{
We   recall   \cite{BULL}, which derives the 
master equation of the reduced density matrix 
for the spin sector of an open Dicke model
 based on approximate dissipation of the harmonic oscillator 
into the Gibbs state. Its (reduced) Liouvillian has dissipative terms 
 having a coefficient corresponding to $\lambda\gamma\delta$ of our case. }}     
\end{rem}

\begin{rem} 
\label{rem:}
{\rm{
Note that  
  $V(\sigma)$ for any  $\sigma > 0$  
 works  for the same transformation identity  (\ref{SPLIT}).   
This freedom  of choice of $\sigma$  reflects the fact that $  - i\mu\K_{\text{sp}}  +  \lambda \gamma\delta\D_{\text{sp}} -i\omega\K_{\text{ph}}   + \gamma\D_{\text{ph}} $ 
in the right-hand side of (\ref{SPLIT}) commutes with $\D_{\text{sp}}$ and 
$\D_{\text{ph}}$.   
From (\ref{V-constans-function}), if $\gamma$ (i.e., 
  the strength of the dissipation of the harmonic oscillator) is large,  
then the strength of the  dissipation of the dressed spin 
in the decoupled Liouvillian (i.e., $\lambda \gamma\delta$) 
  is suppressed.}}   
\end{rem}

\bigskip
In the rest of this section, we suggest a straightforward  
generalization of our results.  
So far, we have discussed the OISD  
model whose dissipation  is assigned to
the harmonic  oscillator only 
as  in  (\ref{LIOU-OISD}).
As we have anticipated in \S\ref{sec:INT}, 
 we can similarly treat a more general  OISD model, where 
 each of   the harmonic oscillator 
and the 
infinite-component spin has its dissipation term.
Namely, let 
\begin{equation}
   {\L_{\text{OISD}}}(\gamma, \bar{\gamma}) :=  -i\mu\K_{\text{sp}} 
   -i\lambda \K^{\text{int}}
-i\omega\K_{\text{ph}}
+ \gamma\D_{\text{ph}} + \bar{\gamma}\D_{\text{sp}} 
\label{LIOU-OISD-GENERAL}
\end{equation}
with positive constants $\omega, \mu, \gamma, \bar{\gamma} $
and real $\lambda$.   
Correspondingly, we define   the  decoupled Liouvillian as 
\begin{equation}
   {\L}_\text{decoupled}(\gamma, \bar{\gamma}) :=  -i\mu\K_{\text{sp}} 
-i\omega\K_{\text{ph}}
+ \gamma\D_{\text{ph}} + \bar{\gamma} \D_{\text{sp}}.  
\label{DECOUPLED-LIOU-GENERAL}
\end{equation}
Because  $\D_{\text{sp}}$ commutes with any of  
$\K_{\text{sp}}$, $\K^{\text{int}}$, $\K_{\text{ph}}$, 
$\D_{\text{sp}}$, $\D^{\text{int}}_{-+}$, $\D^{\text{int}}_{+-}$, and
$\D_{\text{ph}}$ by  (\ref{KD0}), by  using   the same $V(\sigma)$ as in (\ref{V-DEF})
we obtain 
\begin{equation}
{\L_{\text{OISD}}}(\gamma, \bar{\gamma}) V(\sigma)
    =  V(\sigma) {\L}_\text{decoupled}(\gamma, \lambda \gamma\delta+\bar{\gamma}) \,.
\label{GEN-V-TRANS-SPLIT} 
\end{equation}
From (\ref{GEN-V-TRANS-SPLIT}),
 we see that 
the Jaynes--Cummings interaction in the original Liouvillian is erased,  
whereas extra  dissipation is added  to  the infinite-component spin 
but not to  the harmonic oscillator.   

We can discuss
the asymptotic behavior of 
 the dynamical semigroup generated by 
${\L_{\text{OISD}}}(\gamma, \bar{\gamma})$.
Because of  (\ref{GEN-V-TRANS-SPLIT}),
  we can easily show a similar statement to that given in
Proposition~\ref{prop:AYSMP}.

So far, we have always assumed   $\gamma > 0$.     
Finally, let us briefly discuss   the case of $\gamma = 0$ 
(and $\bar{\gamma} \geqslant 0$).      
By recalling  the  argument that yields   (\ref{GEN-V-TRANS-SPLIT}), 
 for any   $\omega \ne \mu$
we see  
\begin{equation}
{\L_{\text{OISD}}}(0, \bar{\gamma})V(\sigma) 
    = V(\sigma) {\L}_\text{decoupled}(0, \bar{\gamma}) 
\label{OPPO-V-TRANS-SPLIT} 
\end{equation} 
by taking  $ V(\sigma) = W(\lambda/(\mu -\omega)) e^{\sigma\D_{\textrm{ph}}}$.    Thus, by applying  the similarity transformation generated by this $V(\sigma)$
 to ${\L_{\text{OISD}}}(0, \bar{\gamma})$,     
 the  interaction term is erased as before, whereas    
 no dissipation  appears in the harmonic oscillator.  
Hence we do not  expect a similar statement to that 
in Proposition~\ref{prop:AYSMP}.
If  the frequencies are identical, that is,  $\omega = \mu$, then 
 the above  decoupling method   is no longer valid. However, 
 assuming   $\omega = \mu$  from the outset does not seem to be 
 a natural setup.

\bigskip

\bigskip

\noindent
\textbf{\large Acknowledgments}

\noindent
This work was supported by JSPS KAKENHI Grant Number JP17K05272.  
The authors are grateful to Professor Valentin A. 
Zagrebnov for useful discussions and  suggestions.  
This work was done  as part of the project of the Rigaku Yugo 1 Group of Kanazawa University.  

\appendix\section{Appendix}
In the text of the note, there are some identities whose derivations are tedious.  
We show here the details of such derivations to facilitate the reader's 
understanding.  

\bigskip

\noindent{\bf Derivation of (\ref{KK})}: 
\begin{align*}
 [\K_A, \K_B](\rho) & = \K_A\K_B(\rho) - \K_B\K_A(\rho) 
 = [A, [B,\rho] ] -  [B, [A,\rho]]  
\\
  &  = A(B\rho - \rho B) - (B\rho - \rho B)A 
    - B(A\rho - \rho A) + (A\rho - \rho A)B 
\\
 & = AB\rho + \rho BA  - BA\rho - \rho AB
   = [A, B] \rho - \rho [A, B]
\\
  & =  \K_{[A, B]} (\rho) \,,
\end{align*}
\begin{align*}
   [\K_A, & \, \D_{B \circ C}](\rho)  = 
    \K_A\D_{B \circ C}(\rho) - \D_{B \circ C}\K_A(\rho)
\\
  & = A(2B\rho C - CB\rho - \rho CB) - (2B\rho C - CB\rho - \rho CB)A
\\
  & - 2B(A\rho - \rho A) C + CB(A\rho - \rho A) + (A\rho - \rho A)CB
\\
  & = 2[A, B]\rho C + 2B\rho [A, C] - ACB\rho  - A\rho CB + CB\rho A 
\\
  & + \rho CBA + CBA\rho  - CB\rho A + A\rho CB - \rho ACB
\\
  & = 2[A, B]\rho C + 2B\rho [A, C]  - (ACB- CBA)\rho  -  \rho(ACB- CBA)
\\
  & = 2[A, B]\rho C + 2B\rho [A, C]  - ( [A, C]B + C[A, B])\rho  -  
      \rho([A, C]B + C[A, B])
\\
  & = \D_{[A, B] \circ C}(\rho)  + \D_{B \circ [A, C]}(\rho) \,.
\end{align*}

\noindent{\bf Derivation of (\ref{DD})}   

Applying each side of the equality to $\rho \in \T(\H)$, we have
\begin{align}
& [\D_{A \circ B}, \D_{ C\circ D}](\rho) 
\notag \\
& = \D_{A \circ B}\big(\D_{ C\circ D}(\rho)\big)
\notag \\
& - \D_{ C\circ D}\big(\D_{A \circ B}(\rho)\big)
\notag \\
&= 2A\big( 2C\rho D - \{DC, \rho \} \big)B - \{BA, \big( 2C\rho D - \{DC, \rho \} \big) \} 
\notag \\
&-2C\big( 2A\rho B - \{BA, \rho \} \big)D + \{DC, \big( 2A\rho B - \{BA, \rho \} \big) \} 
\notag \\
&= 4AC\rho DB - \underline{2A\{DC, \rho \}B} - 2\{BA, C\rho D\} + \{BA,\{DC, \rho \} \} 
\notag \\
& -4CA\rho BD  +2C\{BA, \rho \}D + \underline{2\{DC, A\rho B\} } - \{DC,\{BA, \rho \} \} 
\notag \\
&= (\{A,C\} + [A, C])\rho (\{B, D\} - [B, D])  
\notag \\
& - (\{A,C\} - [A, C])\rho (\{B, D\} + [B, D])
\notag \\
& + \underline{ 2[DC, A]\rho B } + \underline{2 A\rho [B, DC]}  - 2[BA, C]\rho D - 2 C\rho [D,BA] 
\notag \\
& + BADC \rho +  \rho DCBA -  DCBA \rho -  \rho BADC
\notag \\
&=  2 [A, C]\rho \{B, D\} -2 \{A,C\}\rho [B, D] 
\notag \\
&  + 2[DC, A]\rho B + 2 A\rho [B, DC]  - 2[BA, C]\rho D - 2 C\rho [D,BA] 
\notag \\
&  + [[BA,DC], \rho]\,,
\label{left}
\end{align}
where the underlined terms in the fourth expression correspond to 
those in the fifth expression, 
and
\begin{align}
&\big( \D_{[A,C] \circ\{B,D\}} - \D_{\{A, C\} \circ [B, D]} 
\notag \\
& +  \D_{[DC,A] \circ B}   + \D_{A \circ [B,DC]} 
\notag \\
& - \D_{[BA,C] \circ D} - \D_{ C\circ[D,BA] } 
\notag \\
& + \K_{[BA, DC]}\big)(\rho)
\notag \\
& = 2[A,C]\rho \{B,D\} - \{\{B,D\} [A,C], \rho \} - 2\{A,C\}\rho[B, D] + \{[B, D]\{A, C\}, \rho \} 
\notag \\
& + 2 [DC,A]\rho B - \{B[DC,A] , \rho \}  + 2 A\rho[B,DC]  - \{[B,DC]A, \rho \} 
\notag \\
&  - 2 [BA,C]\rho D + \{D[BA,C] , \rho \}  - 2 C\rho[D,BA]  + \{[D,BA]C, \rho \} 
\notag \\
& + [[BA,DC], \rho] \,.
\label{right}
\end{align}
Then the difference of  both sides is
\begin{align}
& \mbox{(\ref{right})} \ - \ \mbox{(\ref{left})}
\notag \\
& = - \{\{B,D\} [A,C], \rho \}   + \{[B, D]\{A, C\}, \rho \} 
\notag \\
&  - \{B[DC,A] , \rho \}    - \{[B,DC]A, \rho \} 
\notag \\
&  + \{D[BA,C] , \rho \}   + \{[D,BA]C, \rho \} 
\notag \\
&=  \{ - \{B,D\} [A,C] + [B, D]\{A, C\} 
\notag \\
&  - B[DC,A]  - [B,DC]A  + D[BA,C] + [D,BA]C, \rho \} 
\notag \\
&= \{ 2BDCA -2DBAC  - BDCA + BADC - BDCA + DCBA 
\notag \\
& \quad + DBAC - DCBA + DBAC - BADC, \rho \}
\notag \\
&= \{     + BADC  + DCBA - DCBA  - BADC, \rho \} = 0 \,.
\end{align}
Thus we have shown  (\ref{DD}). 

\bigskip

\noindent{\bf Derivation of  the first equation in (\ref{3eq})}

By differentiating (\ref{SEMI-aadag}) with respect to $t$, we have
\begin{align*}
   \frac{d}{dt} S_t(\rho)  = &
   \sum_{n=1}^{\infty}\frac{(1-e^{-2t})^{n-1}}{(n-1)!}2e^{-2t}
   a^{\dagger n}e^{-taa^{\dagger}} \rho \, 
   e^{-taa^{\dagger}} a^n  
\\
  &+  \sum_{n=0}^{\infty}\frac{(1-e^{-2t})^n}{n!}
   a^{\dagger n} \big\{ -aa^{\dagger}, \, e^{-taa^{\dagger}} \rho \, 
   e^{-taa^{\dagger}} \big\} a^n  
\\
  = &  \sum_{n=1}^{\infty}\frac{(1-e^{-2t})^{n-1}}{(n-1)!}2e^{-2t}
   a^{\dagger n}e^{-taa^{\dagger}} \rho \, 
   e^{-taa^{\dagger}} a^n  
\\
   & +  \sum_{n=0}^{\infty}\frac{(1-e^{-2t})^n}{n!}2n
   a^{\dagger n} e^{-taa^{\dagger}} \rho \, 
   e^{-taa^{\dagger}} a^n 
\\
  & +  \sum_{n=0}^{\infty}\frac{(1-e^{-2t})^n}{n!}
   \big\{ -aa^{\dagger}, \, a^{\dagger n}e^{-taa^{\dagger}} \rho \, 
   e^{-taa^{\dagger}} a^n \big\} 
\\
  = &   \sum_{n=1}^{\infty}\frac{(1-e^{-2t})^{n-1}}{(n-1)!}
     \big(2e^{-2t} +2(1-e^{-2t})\big)
   a^{\dagger n} e^{-taa^{\dagger}} \rho \, 
   e^{-taa^{\dagger}} a^n 
\\
  & +  \sum_{n=0}^{\infty}\frac{(1-e^{-2t})^n}{n!}
   \big\{ -aa^{\dagger}, \, a^{\dagger n}e^{-taa^{\dagger}} \rho \, 
   e^{-taa^{\dagger}} a^n \big\} 
\\
  = & 2a^{\dagger } \Big( \sum_{n=0}^{\infty}\frac{(1-e^{-2t})^n}{n!}
    a^{\dagger n}e^{-taa^{\dagger}} \rho \, 
   e^{-taa^{\dagger}} a^n \Big) a
\\
  & -\Big\{ aa^{\dagger}, \, \sum_{n=0}^{\infty}\frac{(1-e^{-2t})^n}{n!}
    a^{\dagger n}e^{-taa^{\dagger}} \rho \, 
   e^{-taa^{\dagger}} a^n \Big\}
\\
 = & {\D_{a^{\dagger }\circ a} (S_t(\rho))} \,.
\end{align*}
  
\bigskip

\noindent{\bf Derivation of  the second equation in (\ref{3eq})}:   

By the use of (\ref{SEMI-aadag}) in the right-hand side, we have
\begin{align*}
   S_{t_1}\big(S_{t_2}(\rho)\big) = & 
   \sum_{n=0}^{\infty}  \frac{(1-e^{-2t_1})^n}{n!}a^{\dagger n} e^{-t_1aa^{\dagger}}
   \Big(\sum_{m=0}^{\infty}\frac{(1-e^{-2t_2})^m}{m!}
    a^{\dagger m} e^{-t_2aa^{\dagger}} \rho \, 
   e^{-t_2aa^{\dagger}} a^m \Big) e^{-t_1aa^{\dagger}} a^n  
\\
    = &  \sum_{n=0}^{\infty} \sum_{m=0}^{\infty} 
          \frac{(1-e^{-2t_1})^n(1-e^{-2t_2})^m}{n!m!}
    a^{\dagger (n+m)} e^{-t_1(aa^{\dagger}+m)} e^{-t_2aa^{\dagger}} \rho \, 
   e^{-t_2aa^{\dagger}}  e^{-t_1(aa^{\dagger} + m)} a^{n+m}  
\\
    = &  \sum_{n, m=0}^{\infty}  
          \frac{(1-e^{-2t_1})^n(1-e^{-2t_2})^me^{-2t_1m}}{n!m!}
    a^{\dagger (n+m)} e^{-(t_1+t_2)aa^{\dagger}} \rho \, 
    e^{-(t_1 + t_2)aa^{\dagger}} a^{n+m}  
\\
    = &  \sum_{k=0}^{\infty} \sum_{n=0}^{k} \frac{1}{k!} 
          \frac{k!(1-e^{-2t_1})^n\big(e^{-2t_1} - e^{-2(t_1 + t_2)}\big)^{k-n}}{n!(k-n)!}
    a^{\dagger k} e^{-(t_1+t_2)aa^{\dagger}} \rho \, 
    e^{-(t_1 + t_2)aa^{\dagger}} a^{k}  
\\
    = &  \sum_{k=0}^{\infty}   
          \frac{\big(1-  e^{-2(t_1 + t_2)}\big)^{k}}{k!}
    a^{\dagger k} e^{-(t_1+t_2)aa^{\dagger}} \rho \, 
    e^{-(t_1 + t_2)aa^{\dagger}} a^{k} 
     =   S_{t_1 + t_2}(\rho) \,. 
\end{align*}

\bigskip

\noindent{\bf Derivation of (\ref{trace=1})}: 

Making use of (\ref{hobase}) and (\ref{CPgutai}), we see the action of the 
left-hand side of (\ref{trace=1}) on $| m \rangle $.  
\begin{align*}
   \sum_{n=0}^{\infty} E_n(t)^{\dagger}E_n(t) |m\rangle
     & =  \sum_{n=0}^{\infty}\frac{(1-e^{-2t})^n}{n!} e^{-taa^{\dagger}} a^n 
   a^{\dagger n} e^{-taa^{\dagger}}   |k\rangle
\\
   & = \sum_{n=0}^{\infty}\frac{(n+m)!}{n! \, m!} (1-e^{-2t})^n
      e^{-2t(m+1)}   |m \rangle
\\
  & =  \frac{e^{-2t(m+1)}}{(1 - 1 + e^{-2t})^{m+1}}  \, |m \rangle  =   |m \rangle  \,.  
\end{align*} 
Here we have used the Maclaurin expansion formula
\begin{equation}
        \frac{1}{(1-x)^{m+1}} = \sum_{n=0}^{\infty}
   \frac{(n+m)!}{n! \, m!}x^n  \qquad (|x| < 1 )   
\end{equation}
in the fourth equality.   
Because $\{ \, | m \rangle \}_{ m \in \N \cup \{ 0 \}} $ forms C.O.N.S.\ in $\F$, 
 we obtain  (\ref{trace=1}).      

\bigskip

\noindent{\bf For the semigroup property of $e^{t\D_{a\circ a^{\dagger}}}$, we put}   
\begin{equation}
   \tilde{S}_t(\rho) = \sum_{n=0}^{\infty} \tilde{E_n}(t) \rho \tilde{E_n}(t)^{\dagger}  \,,
\end{equation}
where   
\begin{equation}
         \tilde{E_n}(t) = \frac{(e^{2t}-1)^{n/2}}{\sqrt{n!}} 
           a^n e^{-ta^{\dagger}a}  \,.  
\end{equation}

\bigskip

\noindent{\bf Derivation of 
$ \dfrac{d}{dt}\tilde{S}_t  = \D_{a \circ a^{\dagger }} \tilde{S}_t$}: 
\begin{align*}
\frac{d}{dt}\tilde{S}_t(\rho)  = &
   \sum_{n=1}^{\infty}\frac{(e^{2t}-1)^{n-1}}{(n-1)!}2e^{2t}
   a^ne^{-ta^{\dagger}a} \rho \, 
   e^{-ta^{\dagger}a} a^{\dagger n}  
\\
  &+  \sum_{n=0}^{\infty}\frac{(e^{2t}-1)^n}{n!}
   a^{n}\{ -a^{\dagger}a, e^{-ta^{\dagger}a} \rho \, 
   e^{-ta^{\dagger}a}\} a^{\dagger n}  
\\
  = &  \sum_{n=1}^{\infty}\frac{(e^{2t}-1)^{n-1}}{(n-1)!}2e^{2t}
   a^{n}e^{-ta^{\dagger}a} \rho \, 
   e^{-ta^{\dagger}a} a^{\dagger n}  
\\
   & -  \sum_{n=0}^{\infty}\frac{(e^{2t}-1)^n}{n!}2n
   a^{n} e^{-ta^{\dagger}a} \rho \, 
   e^{-ta^{\dagger}a} a^{\dagger n} 
\\
  & +  \sum_{n=0}^{\infty}\frac{(e^{2t}-1)^n}{n!}
   \{ -a^{\dagger}a, a^{n}e^{-ta^{\dagger}a} \rho \, 
   e^{-ta^{\dagger}a} a^{\dagger n} \} 
\\
  = &   \sum_{n=1}^{\infty}\frac{(e^{2t}-1)^{n-1}}{(n-1)!}
     \big(2e^{2t} - 2(e^{2t}-1)\big)
   a^{n} e^{-ta^{\dagger}a} \rho \, 
   e^{-ta^{\dagger}a} a^{\dagger n} 
\\
  & +  \sum_{n=0}^{\infty}\frac{(e^{2t}-1)^n}{n!}
   \{ -a^{\dagger}a, a^{n}e^{-ta^{\dagger}a} \rho \, 
   e^{-ta^{\dagger}a} a^{\dagger n} \} 
\\
  = & 2a \big( \sum_{n=0}^{\infty}\frac{(e^{2t}-1)^n}{n!}
    a^{n}e^{-ta^{\dagger}a} \rho \, 
   e^{-ta^{\dagger}a} a^{\dagger n} \big) a^{\dagger}
\\
  & -\{ a^{\dagger}a, \sum_{n=0}^{\infty}\frac{(e^{2t}-1)^n}{n!}
    a^{n}e^{-ta^{\dagger}a} \rho \, 
   e^{-ta^{\dagger}a} a^{\dagger n} \}
\\
 = & \D_{a \circ a^{\dagger }} \tilde{S}_t(\rho) \,.
\end{align*}
   
\bigskip

\noindent{\bf Derivation of $ \tilde{S}_{t_1}\tilde{S}_{t_2} = \tilde{S}_{t_1 + t_2}$}: 
\begin{align*}
   \tilde{S}_{t_1}\big(\tilde{S}_{t_2}(\rho)\big) = & 
   \sum_{n=0}^{\infty} \sum_{m=0}^{\infty} \frac{(e^{2t_1}-1)^n}{n!}
   \frac{(e^{2t_2}-1)^m}{m!}
   a^{n} e^{-t_1a^{\dagger}a} a^{m} e^{-t_2a^{\dagger}a} \rho \, 
   e^{-t_2a^{\dagger}a} a^{\dagger m}  e^{-t_1a^{\dagger}a} a^{\dagger n}  
\\
    = &  \sum_{n=0}^{\infty} \sum_{m=0}^{\infty} 
          \frac{(e^{2t_1}-1)^n(e^{2t_2}-1)^m}{n!m!}
    a^{n+m} e^{-t_1(a^{\dagger}a-m)} e^{-t_2a^{\dagger}a} \rho \, 
   e^{-t_2a^{\dagger}a}  e^{-t_1(a^{\dagger}a-m)} a^{\dagger (n+m)}  
\\
    = &  \sum_{n, m=0}^{\infty}  
          \frac{(e^{2t_1}-1)^n(e^{2t_2}-1)^me^{2t_1m}}{n!m!}
    a^{n+m} e^{-(t_1+t_2)a^{\dagger}a} \rho \, 
    e^{-(t_1 + t_2)a^{\dagger}a} a^{\dagger (n+m)}  
\\
    = &  \sum_{k=0}^{\infty} \sum_{n=0}^{k} \frac{1}{k!} 
          \frac{k!(e^{2t_1}-1)^n\big( e^{2(t_1 + t_2)} - e^{2t_1}\big)^{k-n}}{n!(k-n)!}
    a^{k} e^{-(t_1+t_2)a^{\dagger}a} \rho \, 
    e^{-(t_1 + t_2)a^{\dagger}a} a^{\dagger k}  
\\
    = &  \sum_{k=0}^{\infty}   
          \frac{\big( e^{2(t_1 + t_2)}-1 \big)^{k}}{k!}
    a^k e^{-(t_1+t_2)a^{\dagger}a} \rho \, 
    e^{-(t_1 + t_2)a^{\dagger}a} a^{\dagger k}  
     =  \tilde{S}_{t_1 + t_2}(\rho)  \,.
\end{align*} 

\bigskip

\noindent{\bf Derivation of $\sum_n\tilde{E_n}\tilde{E_n}^{\dagger} = \bf{1}$} by
applying it on each $| m \rangle \ ( m \in \N \cup\{ 0 \} ) $:
\begin{align*}
   \sum_{n=0}^{\infty} \tilde{E}_n(t)^{\dagger}\tilde{E}_n(t) |m\rangle
     & =  \sum_{n=0}^{\infty}\frac{(e^{2t}-1)^n}{n!} e^{-ta^{\dagger}a} a^{\dagger n} 
   a^n e^{-ta^{\dagger}a}   |m\rangle
\\
  & = \sum_{n=0}^{m}\frac{(e^{2t}-1)^n}{n!}\frac{m!}{(m-n)!} 
    e^{-2tm}  |m \rangle
\\
  & =   (1 + e^{2t}-1)^m e^{-2tm}   |m \rangle = |m \rangle \,.
\end{align*} 

\bigskip

\noindent{\bf Derivation of (\ref{SPLIT})}:  
We show 
\begin{equation}
    e^{\sigma\D_{\text{ph}}}W(\zeta_2)  \Bigl(  - i\mu\K_{\text{sp}}    
    +  \lambda \gamma\delta\D_{\text{sp}}  
    -i\omega\K_{\text{ph}}  + \gamma\D_{\text{ph}}  \Bigr) W( - \zeta_2)   
\end{equation}
\begin{equation}
  - W( - \zeta_1) \Bigl(  -i\mu\K_{\text{sp}}   -i\lambda \K^{\text{int}}  
-i\omega\K_{\text{ph}} 
+ \gamma\D_{\text{ph}}\Bigr) W(\zeta_1) e^{\sigma\D_{\text{ph}}} = 0 \,. 
\end{equation}
By the use of (\ref{vecCR1}), (\ref{vecCR2}), and part of (\ref{KD0}): 
\begin{equation}
 [\D_{\text{ph}}, \K_{\text{ph}}] = [\D_{\text{ph}}, \K_{\text{sp}}] 
      = [\D_{\text{ph}}, \D_{\text{sp}}] = 0 \,, 
\end{equation}
we have
\begin{align*}
   & e^{\sigma\D_{\text{ph}}} W(\zeta_2)  \Bigl[  - i\mu\K_{\text{sp}}    
    +  \lambda \gamma\delta\D_{\text{sp}}  
    -i\omega\K_{\text{ph}}  + \gamma\D_{\text{ph}}  \Bigr] W( - \zeta_2)   
 \\
    &\quad  - W( - \zeta_1) \Bigl[  -i\mu\K_{\text{sp}}   -i\lambda \K^{\text{int}}  
-i\omega\K_{\text{ph}} 
+ \gamma\D_{\text{ph}}\Bigr] W(\zeta_1) e^{\sigma\D_{\text{ph}}} 
\\
  & =  e^{\sigma\D_{\text{ph}}} \Bigl[  - i\mu \big(   \K_{\text{sp}} 
+ \zeta_2 \K^{\text{int}}_{-+} + \bar{\zeta_2}\K^{\text{int}}_{+-} \big)    
    +  \lambda \gamma\delta\D_{\text{sp}}  
    -i\omega \big( \K_{\text{ph}} - \zeta_2 \K^{\text{int}}_{-+} 
- \bar{\zeta_2}\K^{\text{int}}_{+-} \big)
\\
& \quad  + \gamma \big(\D_{\text{ph}} - \zeta_2 \D^{\text{int}}_{-+} 
  - \bar{\zeta_2}\D^{\text{int}}_{+-}  + |\zeta_2|^2\D_{\text{sp}} \big) \Bigr]
\\ 
& \quad -  \Bigl[  -i\mu\big( \K_{\text{sp}} - \zeta_1 \K^{\text{int}}_{-+} 
- \bar{\zeta_1}\K^{\text{int}}_{+-}  \big)
 -i\lambda (\K^{\text{int}}_{-+} + \K^{\text{int}}_{+-})  
-i\omega\big( \K_{\text{ph}} + \zeta_1 \K^{\text{int}}_{-+} 
+ \bar{\zeta_1}\K^{\text{int}}_{+-}  \big)
\\ 
  & \quad + \gamma \big( \D_{\text{ph}} + \zeta_1 \D^{\text{int}}_{-+} + \bar{\zeta_1}\D^{\text{int}}_{+-}  
   + |\zeta_1|^2\D_{\text{sp}} \big)
\Bigr] e^{\sigma\D_{\text{ph}}} 
\\
& =  e^{\sigma\D_{\text{ph}}} \Bigl[  - i\mu    \K_{\text{sp}} 
- i(\mu -\omega) \zeta_2 \K^{\text{int}}_{-+}  - i(\mu -\omega)\bar{\zeta_2}\K^{\text{int}}_{+-}     
    +  \lambda \gamma\delta\D_{\text{sp}}  
    -i\omega  \K_{\text{ph}} 
\\
&  \quad + \gamma \big(\D_{\text{ph}} - \zeta_2 \D^{\text{int}}_{-+} 
  - \bar{\zeta_2}\D^{\text{int}}_{+-}  + |\zeta_2|^2\D_{\text{sp}} \big) \Bigr]
\\ &  -  \Bigl[  -i\mu \K_{\text{sp}} + i\big((\mu-\omega) \zeta_1 -\lambda\big)\K^{\text{int}}_{-+} 
+ i\big((\mu-\omega) \bar{\zeta_1}-\lambda\big)\K^{\text{int}}_{+-}  
-i\omega \K_{\text{ph}} 
\\ 
  & \quad + \gamma \big( \D_{\text{ph}} + \zeta_1 \D^{\text{int}}_{-+} + \bar{\zeta_1}\D^{\text{int}}_{+-}  
   + |\zeta_1|^2\D_{\text{sp}} \big)
\Bigr] e^{\sigma\D_{\text{ph}}} 
\\
 & =  \Bigl[  - i\mu  \K_{\text{sp}} 
- i(\mu -\omega) \zeta_2 \big( \K^{\text{int}}_{-+} \cosh \sigma + \D^{\text{int}}_{-+} \sinh \sigma  \big)  
- i(\mu -\omega)\bar{\zeta_2} \big( \K^{\text{int}}_{+-}\cosh \sigma  - \D^{\text{int}}_{+-} \sinh \sigma  \big)
\\
 & \quad 
    +  \lambda \gamma\delta\D_{\text{sp}}    -i\omega  \K_{\text{ph}} 
\\
 & \quad + \gamma \big(\D_{\text{ph}} - \zeta_2 \big( \D^{\text{int}}_{-+}\cosh \sigma  + \K^{\text{int}}_{-+} \sinh \sigma \big) 
  - \bar{\zeta_2} \big( \D^{\text{int}}_{+-}\cosh \sigma  - \K^{\text{int}}_{+-} \sinh \sigma \big)
 + |\zeta_2|^2\D_{\text{sp}} \big) \Bigr]e^{\sigma\D_{\text{ph}}} 
\\ 
&  \quad -  \Bigl[  -i\mu \K_{\text{sp}} + i\big((\mu-\omega) \zeta_1 -\lambda\big)\K^{\text{int}}_{-+} 
+ i\big((\mu-\omega) \bar{\zeta_1}-\lambda\big)\K^{\text{int}}_{+-}  
-i\omega \K_{\text{ph}} 
\\ 
  & \quad + \gamma \big( \D_{\text{ph}} + \zeta_1 \D^{\text{int}}_{-+} + \bar{\zeta_1}\D^{\text{int}}_{+-}  
   + |\zeta_1|^2\D_{\text{sp}} \big)
\Bigr] e^{\sigma\D_{\text{ph}}} 
\end{align*}

\begin{align*}
& =  \Bigl[  
\Big(- i(\mu -\omega) \zeta_2  \cosh \sigma 
- \gamma  \zeta_2  \sinh \sigma  
- i\big((\mu-\omega) \zeta_1 -\lambda\big) \Big)\K^{\text{int}}_{-+} 
\\
& \quad +\Big(- i(\mu -\omega)\bar{\zeta_2} \cosh \sigma  
+ \gamma\bar{\zeta_2} \sinh \sigma 
- i\big((\mu-\omega) \bar{\zeta_1}-\lambda\big)\Big)\K^{\text{int}}_{+-}  
\\
& \quad +\Big(- i(\mu -\omega) \zeta_2 \sinh \sigma   
-\gamma  \zeta_2 \cosh \sigma   
- \gamma   \zeta_1\Big) \D^{\text{int}}_{-+}
\\
& \quad +\Big( i(\mu -\omega)\bar{\zeta_2} \sinh \sigma  
- \gamma\bar{\zeta_2} \cosh \sigma  
 - \gamma\bar{\zeta_1}\Big) \D^{\text{int}}_{+-}  
\\
 & \quad 
    + \Big( \lambda \gamma\delta  
 + \gamma|\zeta_2|^2 
 -\gamma |\zeta_1|^2\Big)\D_{\text{sp}} 
\Bigr] e^{\sigma\D_{\text{ph}}} 
\\
& =  \Bigl[  
\Big(- i(\mu -\omega) (\zeta_1+\zeta_2  \cosh \sigma) 
-i \gamma^2\delta 
+ i\lambda \Big)\K^{\text{int}}_{-+} 
\\
& \quad +\Big(- i(\mu -\omega)( \bar{\zeta_1} +  \bar{\zeta_2} \cosh \sigma)  
-i\gamma^2\delta + i\lambda\Big)\K^{\text{int}}_{+-}  
\\
& \quad +\Big((\mu -\omega)\gamma\delta   
-i \gamma^2\delta  \coth \sigma   
- \gamma   \zeta_1\Big) \D^{\text{int}}_{-+}   
\\
& \quad +\Big( (\mu -\omega)\gamma\delta  
+i \gamma^2\delta \coth \sigma  
 - \gamma\bar{\zeta_1}\Big) \D^{\text{int}}_{+-}  
\\
 & \quad 
    + \Big( \lambda \gamma\delta  
 + \gamma|\zeta_2|^2 
 -\gamma |\zeta_1|^2\Big)\D_{\text{sp}} 
\Bigr] e^{\sigma\D_{\text{ph}}} \qquad = 0 \,.
\end{align*}
We have used (\ref{V-constans-function}) in the last two steps.   


\end{document}